\documentclass[11pt]{article}

\usepackage{amsthm,amsmath,amssymb}
\usepackage[noend]{algorithmic}
\usepackage{algorithm}
\usepackage{epsfig, graphicx}
\usepackage[left=1in,top=1in,right=1in,bottom=1in,nohead]{geometry}
\usepackage{times}

\newtheorem{lemma}{Lemma}

\newtheorem{theorem}[lemma]{Theorem}
\newtheorem{corollary}[lemma]{Corollary}
\newtheorem{definition}[lemma]{Definition}
\newenvironment{LabeledProof}[1]{\noindent{\bf Proof of #1: }}{\qed}

\newcommand{\partn}{{\cal P}}
\newcommand{\hier}{{\cal H}}
\newcommand{\Diam}[1]{\mbox{\sc Diam}(#1)}
\newcommand{\Opt}[1]{\mbox{\sc Opt}(#1)}
\newcommand{\cost}[1]{\mbox{\sc Cost}(#1)}
\newcommand{\dist}[2]{d(#1,#2)}
\newcommand{\TreeDist}[3]{d_{#1}(#2,#3)}

\newcommand{\ust}{\mbox{UST}}

\newcommand{\subG}[2]{#1[#2]}
\newcommand{\superG}[2]{\widehat{#1}[#2]}
\newcommand{\subH}[2]{#1[#2]}
\newcommand{\subP}[2]{#1[#2]}
\newcommand{\pn}{{\cal P}}
\newcommand{\inE}[2]{\mbox{{\rm in}}_{#1}(#2)}
\newcommand{\outE}[2]{\mbox{{\rm out}}_{#1}(#2)}
\newcommand{\hH}{\widehat{D}}

\newcommand{\cntr}{\mbox{{\sc dest}}}
\newcommand{\newpartn}{{\cal Q}}
\newcommand{\detour}[2]{\mbox{{\sc dtr}}_{#1}(#2)}
\newcommand{\MaxDiam}{\mbox{{\sc MaxDiam}}}

\newcommand{\Real}{{\mathbb R}}
\newcommand{\BfPara}[1]{{\noindent {\bf #1.}}}
\newcommand{\polylog}{\mbox{polylog}}

\newcounter{listCounter}

\newcommand{\ListLengths}{\setlength{\itemsep}{0ex}\setlength{\topsep}{1ex}\setlength{\partopsep}{0ex}}

\newenvironment{NoIndentEnumerate}{\begin{list}{\arabic{listCounter}.}{
      \usecounter{listCounter}\setlength{\leftmargin}{1em}\ListLengths}}{\end{list}}

\newenvironment{mylowitemize}{\begin{list}{$\bullet$}{\setlength{\leftmargin}{1em}\ListLengths}}{\end{list}}

\newcommand{\junk}[1]{}

\begin{document}

\title{On Strong Graph Partitions and Universal Steiner Trees\footnote{A preliminary version of this work appeared 
in the Proceedings of the 53rd Annual IEEE Symposium on Foundations of Computer Science, 2012, pp. 81-90.}}

\author{Costas Busch \thanks{Department of Computer Science, Louisiana State University, Baton Rouge, USA. 
  E-mail: {\tt busch@csc.lsu.edu}.} 
  \and 
  Chinmoy Dutta \thanks{Twitter Inc., San Francisco, USA.
  E-mail: {\tt chinmoy@twitter.com}. 
  Chinmoy Dutta was at the College of Computer and Information Science, Northeastern Univrsity, and was supported in
  part by NSF grant CCF-0845003 and a Microsoft grant to Ravi Sundaram during this work.} 
  \and 
  Jaikumar Radhakrishnan \thanks{School of Technology and Computer Science, Tata Institute of Fundamental Research, Mumbai, India.  
  E-mail: {\tt jaikumar@tifr.res.in}.} 
  \and 
  Rajmohan Rajaraman \thanks{College of Computer and Information Science, Northeastern University, Boston, USA.  
  E-mail: {\tt rraj@ccs.neu.edu}.
  Rajmohan Rajaraman is supported in part by NSF grant CNS-0915985.} 
  \and
  Srivathsan Srinivasagopalan \thanks{CTS / VISA, Mountain View, USA. 
  E-mail: {\tt ssrini1@csc.lsu.edu}.
  Srivathsan Srinivasagopalan was at the Department of Computer Science, Louisiana State University, at the time of this work.}
}

\maketitle

\begin{abstract}
We study the problem of constructing universal Steiner trees for
undirected graphs. Given a graph $G$ and a root node $r$, we seek a
single spanning tree $T$ of minimum {\em stretch}, where the stretch
of $T$ is defined to be the maximum ratio, over all terminal sets $X$,
of the cost of the minimal sub-tree $T_X$ of $T$ that connects $X$ to
$r$ to the cost of an optimal Steiner tree connecting $X$ to $r$ in
$G$. Universal Steiner trees (USTs) are important for data aggregation
problems where computing the Steiner tree from scratch for every input
instance of terminals is costly, as for example in low energy sensor
network applications.

We provide a polynomial time \ust\ construction for general graphs
with $2^{O(\sqrt{\log n})}$-stretch. We also give a polynomial time
$\polylog(n)$-stretch construction for minor-free graphs. One basic
building block of our algorithms is a hierarchy of graph partitions,
each of which guarantees small strong diameter for each cluster and
bounded neighbourhood intersections for each node. We show close
connections between the problems of constructing USTs and building
such graph partitions. Our construction of partition hierarchies for
general graphs is based on an iterative cluster merging procedure,
while the one for minor-free graphs is based on a separator theorem
for such graphs and the solution to a cluster aggregation problem that
may be of independent interest even for general graphs. To our
knowledge, this is the first subpolynomial-stretch ($o(n^\epsilon)$
for any $\epsilon > 0$) UST construction for general graphs, and the
first polylogarithmic-stretch UST construction for minor-free graphs.
\end{abstract}

\thispagestyle{empty}
\newpage
\setcounter{page}{1}


\section{Introduction}
\label{sec:intro}
In this paper, we study universal approximations for the Steiner Tree
problem on undirected graphs. In the universal Steiner Tree (\ust)
problem for graphs, we are given an undirected graph $G$ and a
designated root vertex $r$ in $G$, and the task is to find a {\em
  single spanning tree} $T$ of $G$ such that for any set $X$ of
terminal vertices, the minimal subtree $T_X$ of $T$ that connects $X$
to $r$ is a good approximation to the optimal Steiner tree connecting
$X$ to $r$ in $G$.  The quality of the solution $T$ is given by its
{\em stretch}, which is the maximum ratio of the cost of $T_X$ to the
cost of the optimal Steiner tree connecting $X$ to $r$ in $G$ over all
terminal sets $X$.

The universal Steiner tree problem has been studied extensively for
the case of metrics where one is allowed to output an ``overlay
tree'', whose edges correspond to paths in the given
graph~\cite{jia+lnrs:ust,gupta+hr:oblivious,bhalgat+ck:differential,srinivasagopalan+bi:dd}.
Equivalently, the case of metrics can be viewed as a complete graph in
which all edge weights satisfy the triangle inequality. In fact, for
the case of metrics, there have been several important results on
extensions of the UST problem and variants seeking sparse network
structures that simultaneously approximate the optimal solutions for a
range of input
instances~\cite{goel+e:simultaneous,goel+p:constant,goel+p:oblivious,gupta+hr:oblivious}.

The focus of this paper is on the \ust\ problem on arbitrary graphs
where we require that the solution being sought is a spanning tree of
the given graph. The Minimum Steiner tree problem on a graph can be
well-approximated by solving the same problem on the metric induced by
the graph and then computing the minimum subtree connecting the
terminals. Such an approach, however, does not apply to the
\ust\ problem owing to the requirement that the tree {\em
  simultaneously} approximate the optimal Steiner tree for {\em all}
terminal sets. Note that this is a much stronger requirement than
asking for a probability distribution over spanning trees that has
small expected stretch for every terminal set. In the latter case,
there might not be any single tree in the distribution that is good
for all terminal sets, i.e., for every tree there is a terminal set
such that the minimal subtree connecting the terminals to the root has
a cost much larger than the optimal steiner tree.

\smallskip
\BfPara{Motivation} Our problem formulation is primarily motivated by
information aggregation and data dissemination in sensor and ad-hoc
wireless
networks~\cite{madden+fhh:tag,madden+sfc:aggrequery,krishnamachari+ew:datacentric}.
In a sensor network, data is often collected by a central agent that
periodically queries a subset of sensors for their sensed
information. In many applications, the queries seek aggregate
information which can be transmitted using a low cost tree that
aggregates data at intermediate nodes.  This reduces the number of
transmissions which is crucial as sensors have limited battery life
and wireless transmissions are power intensive.  It is not realistic,
however, to expect the sensors to compute and store a low cost tree
for each potential subset of sensors being aggregated as the sensors
have limited memory and computational power.  In this setting, a
universal tree provides a practical solution where the nodes just need
to realize a single tree which approximates optimal aggregation trees
for all subsets of sensors.  Thus, one natural approach is to employ a
universal {\em overlay}\/ tree.  This has several disadvantages,
however.  First, aggregation over the overlay tree requires a physical
routing infrastructure that supports point-to-point communication
among distant nodes in the network.  Second, even if such an
infrastructure exists, it may not route packets along minimum-cost
paths as required by the overlay tree.  Furthermore, aggregation over
the overlay tree requires synchronization among distant nodes in the
network and incurs overhead in terms of delays and storage.  Thus, in
some resource-constrained applications, we would ideally want to
construct a universal spanning tree as opposed to an overlay tree.

Another motivation to study universal approximation algorithms comes
from their relation with differential privacy which was recently
established by Bhalgat, Chakrabarty and Khanna
\cite{bhalgat+ck:differential}.  They showed that universal solutions
such as \ust s are differentially private, and argued that a kind of
``strong'' lower bounds for universal algorithms implies lower bounds
for differentially private ones as well.

From a theoretical standpoint, our motivation is to find out whether
the results known for \ust\ and related problems in the metric case
can, in fact, be achieved using spanning trees of the underlying
graphs. The analogous question for approximating metrics by tree
metrics has been answered affirmatively
by~\cite{elkin+est:stretch,abraham+n:petal,abraham+bn:stretch} who showed that
nearly logarithmic-stretch spanning trees exist for all graphs, almost
matching the best bound achievable by tree
metrics~\cite{fakcharoenphol+rt:treemetric}. No comparable results are
known for the \ust\ problem.

\subsection{Our results and techniques}
\label{sec:results}
Our main results are \ust\ algorithms for general graphs and for the
special class of minor-free graphs.  
\begin{mylowitemize}
\item
{\bf \ust\ for general graphs:} We present a polynomial-time algorithm
for computing a $2^{O(\sqrt{\log n})}$-stretch spanning tree for any
undirected graph.

\item
{\bf \ust\ for minor-free graphs:} We present a polynomial-time
algorithm for computing a $\polylog(n)$-stretch spanning tree for
any graph that is $H$-minor free for any finite graph $H$.
\end{mylowitemize}
While the specific techniques used in the two algorithms are
substantially different, both are grounded in a common general
framework that draws close connections between \ust s and certain
graph partitions based on strong diameter.  We define an $(\alpha,
\beta, \gamma)$-partition of a graph $G$ as a partition of the
vertices of $G$ into clusters such that each cluster has strong
diameter at most $\alpha \gamma$, and for every vertex, the ball of
radius $\gamma$ in $G$ intersects at most $\beta$ clusters.  A primary
motivation to study these partitions is the following result.
\begin{mylowitemize}
\item
{\bf From \ust s to partitions:} If every $n$-vertex graph has a
$\sigma(n)$-stretch \ust\ for some function $\sigma$, then for any
real $\gamma > 0$, every $n$-vertex graph has an $(O(\sigma(n)^2),
O(\sigma(n)), \gamma)$-partition.  Moreover, such a partition can be
efficiently constructed given black-box access to a
$\sigma(n)$-stretch \ust\ algorithm. (Section~\ref{sec:ust_to_partn})
\end{mylowitemize}
While the above result says that one cannot construct \ust s without
(implicitly) constructing these graph partitions, the significance of
our framework stems from our next result that one can also efficiently
construct \ust s from these strong partitions.  We define an $(\alpha,
\beta, \gamma)$-partition hierarchy as a sequence of partitions
starting from the trivial partition in which each vertex forms its own
cluster, and the $i$th partition is an $(\alpha, \beta,
\gamma^i)$-partition that coarsens the $(i-1)$th partition. (See
Section~\ref{sec:defn} for formal definitions.)  Given a partition
hierarchy, a natural divide-and-conquer method to construct a
\ust\ (similar to one employed in~\cite{jia+lnrs:ust} for metric \ust)
is to connect together subtrees recursively computed for lower levels
of the hierarchy.  This approach, however, does not work.  In fact, we
prove that {\em any}\/ \ust\ construction that {\em strictly
  obeys}\/ the connectivity structure of the hierarchy, in the sense
that the subgraph of the tree induced by every cluster of the
hierarchy is connected, will have poor stretch in the worst case (see
Section~\ref{sec:bottom-up}).  We overcome this obstacle by
introducing the novel notion of spanning trees that {\em approximately
  respect}\/ a given partition hierarchy; such a tree may be
disconnected within a cluster of the hierarchy, but is joined
externally so as to approximately respect the distances within every
cluster.  We show how to construct such spanning trees from a given
partition hierarchy and prove that they achieve desired stretch
factors.
\begin{mylowitemize}
\item
{\bf From partition hierarchies to \ust s:} For any graph $G$, given
an $(\alpha, \beta, \gamma)$-partition hierarchy for $G$, an
$O(\alpha^2 \beta^2 \gamma \log n)$-stretch \ust\ for $G$ can
be constructed in polynomial time. (Section~\ref{sec:partn_to_ust})
\end{mylowitemize}
A major consequence of the above result is that one can obtain a
$\polylog(n)$-stretch \ust\ by constructing a $(\polylog(n),
\polylog(n), \polylog(n))$-partition hierarchy.  Note that there is an
$\Omega(\log n)$ lower bound on the best stretch achievable, even in
the metric case~\cite{gupta+hr:oblivious,bhalgat+ck:differential}. We
next obtain our main results for general graphs and minor-free graphs
by constructing suitable partition hierarchies.
\begin{mylowitemize}
\item
{\bf Partition hierarchies for general graphs:} Every graph $G$ has a
polynomial-time computable\\ $(2^{O(\sqrt{\log n})}, 2^{O(\sqrt{\log
    n})}, 2^{O(\sqrt{\log n})})$-partition hierarchy.
(Section~\ref{sec:hierarchical})

\item
{\bf Partition hierarchies for minor-free graphs:} Every minor-free
graph $G$ has a polynomial-time computable $(O(\log^3 n), O(\log^4 n),
O(\log^3 n))$-partition hierarchy. (Section~\ref{section:minor-free})
\end{mylowitemize}
The partition hierarchy for general graphs is obtained by an iterative
procedure in which clusters are continually merged by identifying
vertices for which the number of intersecting clusters within a
specified distance exceeds the desired bound.  The particular order in
which the vertices are processed is carefully chosen; a natural greedy
approach fails.

Our construction of the partition hierarchy for minor-free graphs is
more complicated.  It is based on a separator theorem due
to~\cite{thorup:planar,abraham+g:separators} which builds
on~\cite{klein+pr:minors} and shows that any minor-free graph can be
decomposed into connected components, each of which contains at most
half the number of nodes, by removing a sequence of a constant number
of shortest paths.  A key ingredient of our hierarchical construction
for minor-free graphs is a result on cluster aggregation in general
graphs, which is of independent interest.

\begin{mylowitemize}
\item
{\bf Cluster aggregation:} We show that given any partition of $G$
into disjoint clusters each with strong diameter at most $D$, and a
set $S$ of portal vertices, we can aggregate the clusters into
disjoint connected components, each component with a distinguished
portal from $S$, such that for any vertex $v$, the distance, within
the component of $v$, from $v$ to the distinguished portal in the
component is at most $O(\log^2 n) D$ more than the distance of $v$ to
$S$ in $G$. (Section~\ref{sec:cluster})
\end{mylowitemize}

\subsection{Related work}
Research in network design over the past decade has revealed that it
is often possible to derive sparse network structures (e.g., routes,
multicast trees) that yield good approximations simultaneously for a
range of input instances.  One of the earliest examples of such a
result is due to Goel and Estrin \cite{goel+e:simultaneous} who
introduced the problem of {\em simultaneous single sink buy-at-bulk}
and gave an $O(\log D)$ bound on the simulataneous ratio where $D$ is
the total demand.  The guarantee is that their solution works
simultaneously for all fusion cost function $f$ which are concave and
monotonically non-decreasing with $f(0) = 0$. In a related paper
\cite{goel+p:oblivious}, Goel and Post constructed a distribution over
trees such that the expected cost of a tree for any $f$ is within an
$O(1)$-factor of the optimum cost for that $f$. A recent improvement
by Goel and Post \cite{goel+p:constant} provides the first constant
guarantee on the simultaneous ratio achievable by a tree.  This result
is incomparable to our results since the set of terminals that are
being aggregated in the buy-at-bulk problem are fixed.  

Jia et al. \cite{jia+lnrs:ust} introduced the notion of universal
approximation algorithms for optimization problems, focusing on TSP,
Steiner Tree and set cover problems. For the universal Steiner tree
problem, they presented polynomial-time algorithms that construct
overlay trees with a stretch of $O(\log^4 n/\log \log(n))$ for
arbitrary metrics and logarithmic stretch for doubling, Euclidean, or
growth-restricted metrics.  At a high-level, our approach of using
partition hierarchies to derive \ust s is similar to that
of~\cite{jia+lnrs:ust}.  There are several critical differences,
however.  First, as we discussed in Section~\ref{sec:results}, the
natural divide-and-conquer method employed in~\cite{jia+lnrs:ust} of
constructing the \ust\ fails for graphs.  Second, the construction of
{\em strong partitions}\/ for graphs (as opposed to the weak
partitions of~\cite{jia+lnrs:ust}) require entirely new techniques for
both general graphs and minor-free graphs, and introduce new
subproblems of independent interest, e.g., the cluster aggregation
problem studied in Section~\ref{sec:cluster}. Some of these results and techniques appeared in
the preliminary version of this paper~\cite{busch+drrs:ust}. The work
of~\cite{jia+lnrs:ust} also provided a lower bound of $\Omega(\log n/
\log \log n)$ for \ust\ that holds even when all the vertices are on a
plane; for general metrics, this can be improved to $\Omega(\log
n)$~\cite{gupta+hr:oblivious,bhalgat+ck:differential}.  Note that
these lower bounds extend to the \ust\ problem on graphs.  Lower
bounds for universal TSP are given
in~\cite{hajiaghayi+kl:tsp,gorodezky+kss:tsp}.  For earlier work on
universal TSP, see~\cite{platzman+b:tsp,bertsimas+g:euclideantsp}.

Gupta, Hajiaghayi and R\"{a}cke \cite{gupta+hr:oblivious} developed an
elegant framework to model {\em oblivious network design} problems and
gave algorithms with poly-logarithmic approximation ratios. They give
network structures that are simultaneously oblivious to both the link
cost functions (subject to them being drawn from a suitable class) and
traffic demand.  Their algorithms are based on the celebrated tree
metric embeddings of Fakcharoenphol et
al. \cite{fakcharoenphol+rt:treemetric} and hierarchical cut-based
decompostions of R\"acke~\cite{racke:congestion}.  For the
\ust\ problem on metrics, the algorithm of~\cite{gupta+hr:oblivious}
builds a \ust\ as follows: First obtain $O(\log n)$ trees from the
distribution of~\cite{fakcharoenphol+rt:treemetric}; next assign each
non-root vertex to a tree that well-approximates its distances to all
other nodes; finally, take the union, over each of the $O(\log n)$
overlay trees, the subtree of the tree induced by the root and the
vertices assigned to the tree.  The resulting overlay tree is an
$O(\log^2 n)$-stretch \ust.

A potential approach to solving the \ust\ problem on graphs is to
adapt the techniques of~\cite{gupta+hr:oblivious} with $O(\log n)$
spanning trees drawn from the distributions
of~\cite{elkin+est:stretch} instead of the overlay trees
of~\cite{fakcharoenphol+rt:treemetric}.  A major challenge here is
that the paths or subtrees chosen from the different $O(\log n)$ trees
may share several vertices and hence create unavoidable cycles when
combined.  The only prior work on constructing universal Steiner trees
for graphs is due to Busch et al.~\cite{srinivasagopalan+bi:dd} who
achieved a stretch of $O(\log^3 n)$ for the restricted class of graphs
with bounded doubling dimension by showing how one can continually
refine an $O(\log n)$-stretch overlay tree by removing cycles to
obtain an $O(\log^3 n)$-stretch \ust.  Their techniques, however, are
closely tied to the particular class of graphs and seem difficult to
generalize.  We also note that the spanning tree constructions aimed
at minimizing average
stretch~\cite{alon+kpw:stretch,elkin+est:stretch,abraham+bn:stretch,abraham+n:petal}
with respect to distance do not yield any useful bounds for our
stretch measure with respect to optimal Steiner trees.

\junk{oblivious'' (needs knowledge of just the class of the fusion cost
function but not the actual function) and ``traffic oblivious''
(source-demand pairs arrive online and are routed immediately without
any knowledge of the state of the network or other source-demand pairs
present). For the problem of minimizing the total cost of the network,
they gave oblivious algorithm with $O(\log^2 n)$ competitive ratio
whenever the fusion cost function is monotone and sub-additive. 

Their work
provided improved $O(\log^2 n)$ stretch bounds for both universal
Steiner tree and universal TSP on metrics.
}


\junk{The aforementioned problem of approximating a graph metric by a tree
metric has a rich history.  Alon et al. \cite{alon+kpw:stretch} showed
an upper bound of $2^{\sqrt{\log n \log \log n}}$ for approximating an
arbitrary graph metric by a distribution over spanning trees. Bartal
\cite{bartal:treemetrics} showed that an improved $O(\log^2 n)$
approximation is achievable using tree metrics if one drops the
requirement that the trees be subgraphs of the underlying graph.
Konjevod et al. \cite{konjevod+rs:planar} improved Bartal's result to
$O(\log n)$ for planar graphs while Charikar et
al. \cite{charikar+cggp:finite} improved it for low dimensional normed
spaces. Subsequently, Bartal \cite{bartal:improved} improved his
earlier result to $O(\log n \log \log n)$ and also showed a lower
bound of $\Omega(\log n)$ on the distortion for probabilistically
embedding an expander graph into a tree. Fakcharoenphol, Rao and
Talwar \cite{fakcharoenphol+rt:treemetric} closed the gap between the
lower and the upper bound by showing that arbitrary metrics can be
approximated by a distribution over tree metrics with distortion
$O(\log n)$.  More recently, Elkin, Emek, Spielman, and
Teng~\cite{elkin+est:stretch} showed an upper bound of $O(\log^2 n
\log \log n)$ for approximating an arbitrary graph metric using a
distribution of spanning trees, thus significantly improving the
result of Alon et al~\cite{alon+kpw:stretch}.  This result was
subsequently improved by Abraham, Bartal, and Neiman
\cite{abraham+bn:stretch} who achieved an $O(\log n \log \log n(\log
\log \log n)^3)$ bound.
}

As mentioned in Section~\ref{sec:results}, our universal Steiner trees
are based on certain partitions of graphs where we would like to bound
the strong diameter of the clusters while maintaining some sparsity
constraints.  Such partitions have been extensively
studied~\cite{peleg:distributebook,awerbuch+p:sparse}.  While nearly
optimal partitions based on weak diameter bounds are known in many
cases, strong-diameter based decompositions are less
understood~\cite{peleg:distributebook}.  There have been recent
results on strong-diameter
decompositions\cite{abraham+bn:stretch,abraham+gmw:decomposition,abraham+n:petal,elkin+est:stretch};
while our partitions share some of the ideas (e.g., of stitching
together judiciously chosen shortest paths), there are significant
differences in the details and the particular partitions being sought.
In particular, none of the proposed partitions satisfy the requirement
that the neighborhood around every node intersects a small number of
clusters.  Furthermore, while we seek partition hierarchies with
deterministic guarantees, many previous results concerned hierarchies
with either probabilistic or averaging guarantees or covers where
clusters are allowed to overlap.





\section{Definitions and notations}
\label{sec:defn}
Let $G = (V,E,w)$ denote a weighted undirected graph, where $V$ and
$E$ are the sets of vertices and edges, respectively, and $w:E
\rightarrow \Real$ is the length function on edges.  We assume,
without loss of generality, that the minimum edge length is $1$, since
otherwise we can scale all the edge lengths appropriately. The length
of a path is simply the sum of the lengths of the edges in it. For any
$u$ and $v$ in $V$, the distance between $u$ and $v$, denoted by
$\dist{u}{v}$, is the length of a shortest path between $u$ and $v$,
according to $w$. For $v \in V$ and real number $\rho$, let
$B(v,\rho)$ denote the ball of radius $\rho$ centered at $v$, i.e.,
$B(v, \rho)$ is the set of all vertices that are at distance at most
$\rho$ from $v$, including $v$.  The {\em diameter}\/ of the graph,
denoted by $\Diam{G}$, is the maximum distance between any two
vertices of $G$.

For any graph $G$ and any subset $X$ of vertices in $G$, let
$\subG{G}{X}$ denote the subgraph of $G$ induced by $X$.  For any
subset $X$ of vertices and $u$, $v$ in $X$, let
$\TreeDist{X}{u}{v}$ denote the distance between $u$ and $v$ in
$\subG{G}{X}$.

\smallskip
\BfPara{Universal Steiner trees} Given a specified {\em root}\/ vertex
$r \in V$ and a set of {\em terminal}\/ vertices $X \subseteq V$, a
Steiner tree $T$ for $X$ is a minimal subgraph of $G$ that connects
the vertices of $X$ to the root. The {\em cost}\/ of a tree $T$,
denoted by $\cost{T}$, is the sum of the lengths of edges in it.
Assume $G$ and $r$ to be fixed.  We let $\Opt{X}$ denote the cost of
the minimum cost steiner tree connecting $X$ to $r$.  Given a spanning
tree $T$ of $G$ and terminal set $X$, we define its projection on the
terminal set $X$, denoted by $T_X$, as the minimal subtree of $T$
rooted at $r$ that contains $X$.

\begin{definition}[Universal Steiner tree (\ust)]
\label{def:ust}
Let $G$ be an weighted undirected graph, and $r$ be a specified root
vertex in $V$.  We define the {\em stretch}\/ of a spanning tree $T$
of $G$ to be $\max_{X \subseteq V} \cost{T_X}/\Opt{X}$.  The
{\bf universal Steiner tree}\/ problem is to find a spanning tree with
minimum stretch.
\end{definition}

\junk{Note that the requirement that the root be included in the set of
terminals is indispensable, otherwise we cannot have a tree with
better than $\Omega(n)$ stretch which can be easily seen from the
example of a ring.}

\BfPara{Partitions} A {\em partition}\/ $\partn$ of $V$ is a
collection of disjoint subsets of $V$ whose union equals $V$.  We
refer to each element of $\partn$ as a {\em cluster}\/ of the graph
$G$. There are two notions for the diameter of a cluster $C$. This
paper focuses on the {\em strong}\/ diameter, denoted by $\Diam{C}$,
which is the diameter of the subgraph induced by the cluster,
i.e. $\Diam{\subG{G}{C}}$.  In contrast, the {\em weak}\/ diameter of
a cluster is simply the maximum distance between any two verices of
the cluster in $G$.

\begin{definition}[$(\alpha, \beta, \gamma)$-partition]
\label{def:partition}
An $(\alpha, \beta, \gamma)$-partition $\partn$ of $G$ is a partition
of $V$ satisfying:\junk{ the following properties.}
\begin{NoIndentEnumerate}
\item 
{\bf Strong diameter:} The strong diameter of every cluster $C$ in
$\partn$ is at most $\alpha \gamma$; i.e., $\Diam{C} \le \alpha
\gamma$.

\item
{\bf Cluster-valence:} For every vertex $v$ in $V$, $B(v,\gamma)$ has
a nonempty intersection with at most $\beta$ clusters in $\partn$.  We
refer to $\beta$ as the cluster-valence of $\partn$.
\end{NoIndentEnumerate}
\end{definition} 

A notion of partition similar to our $(\alpha,\beta,\gamma)$-partition
appeared in Jia et al. \cite{jia+lnrs:ust}, which required a bound on
the weak diameter of clusters.

\begin{definition}[Partition hierarchy]
\label{def:hierarchical_partition}
For $\gamma > 1$, an $(\alpha, \beta, \gamma)$-partition
hierarchy of a graph $G$ is a sequence $\hier = \langle \partn_0,
\partn_1, \ldots, \partn_d \rangle$ of partitions of $V$, where $d=
\lceil \log_\gamma(\Diam{G}/\alpha) \rceil$, satisfying:\junk{ the
following properties.}
\begin{NoIndentEnumerate}
\item
{\bf Partition:} For $0 \le i \le d$, $\partn_i$ is an $(\alpha,
\beta, \gamma^i)$-partition of $G$.  Furthermore, $\partn_d$ is the
collection $\{V\}$.  For convenience, we set $\partn_{-1}$ to the
collection $\{\{v\} \mid v \in V\}$.

\item
{\bf Hierarchy:} For $0 \le i < d$, every cluster in $\partn_i$ is
contained in some cluster in $\partn_{i+1}$.

\item
{\bf Root Padding:} For $0 \le i \le d$, the ball $B(r,\gamma^i)$ of
radius $\gamma^i$ around root $r$ is contained in some cluster in
$\partn_i$.
\end{NoIndentEnumerate}
\end{definition}

\noindent
For a partition $\partn$ of a graph $G$, let $\superG{G}{\partn}$
denote a weighted graph in which the vertex set is $\partn$, and there
is an edge $(C,C')$ between clusters $C$ and $C'$ if $G$ has an edge
between a vertex in $C$ and a vertex in $C'$; the length of the edge
$(C,C')$ is the minimum length of an edge between $C$ and $C'$
in $G$.

For a partition $\partn$, let $\partn(v)$ denote the cluster of
$\partn$ that contains the vertex $v$ and $\MaxDiam(\partn)$ denote
$\max_{C \in \partn} \Diam{C}$.  For a subset $X$ of vertices, let
$\subP{\partn}{X}$ denote the partition restricted to $X$; i.e.,
$\subP{\partn}{X}$ is the collection $\{X \cap C \mid C \in \partn\}$.
For a partition hierarchy $\hier$ and a cluster $C$ that is an element
of a partition $\partn_i$ in $\hier$, we let $\subH{\hier}{C}$ denote
the partition hierarchy projected to $C$; that is, $ \subH{\hier}{C} =
\langle \subP{\partn_0}{C}, \ldots, \subP{\partn_i}{C} \rangle$.  Let
$T$ be a spanning tree and $\hier$ be an $(\alpha, \beta,
\gamma)$-partition hierarchy of $G$. We say that $T$ {\em strictly
  obeys} $\hier$ if for each $\partn_i \in \hier$ and each cluster $C
\in \partn_i$, the subgraph of $T$ induced by $C$ is connected. We say
that $T$ {\em $\mu$-respects}\/ $\hier$ if for each $\partn_i \in
\hier$, each $C \in \partn_i$ and every pair of vertices $u, v
\in C$, $\TreeDist{T}{u}{v}$ is at most $\mu \alpha \gamma^i$.

\section{Strong partitions and universal Steiner trees}
\label{sec:ust_partn}
We now present close connections between the strong partitions of
Definition~\ref{def:partition} and universal Steiner trees. We first
show in Section~\ref{sec:ust_to_partn} that partitions with low strong
diameter and low cluster-valence are necessary for deriving
low-stretch trees. We next show in Section~\ref{sec:partn_to_ust} how
partition hierarchies yield \ust s.  Given an $(\alpha, \beta,
\gamma)$-partition hierarchy for any graph $G$,
Section~\ref{sec:bottom-up} shows how to get an
$O((\alpha\beta)^{\log_\gamma n} \gamma \beta^2 \log_\gamma
n)$-stretch \ust\ for $G$ that strictly obeys the partition hierarchy,
and also presents a nearly matching lower bound on the stretch for
such \ust s. Section~\ref{sec:top-down} then presents an improved
$O(\alpha^2 \beta^2 \gamma \log_\gamma n)$-stretch \ust\ construction
that does not strictly obey but approximately respects the partition
hierarchy.  All proofs in this section are deferred to
Section~\ref{proofs:ust_partn}.

\subsection{From universal Steiner trees to strong partitions}
\label{sec:ust_to_partn}

\begin{theorem}
\label{thm:ust_2_partn}
If every $n$-vertex graph has a
$\sigma(n)$-stretch \ust\ for some function $\sigma$, then for any
real $\gamma > 0$, every $n$-vertex graph has an $(O(\sigma(n)^2),
O(\sigma(n)), \gamma)$-partition.  Moreover, such a partition can be
efficiently constructed given black-box access to a
$\sigma(n)$-stretch \ust\ algorithm.
\junk{Given an algorithm $\mathcal{A}$ to construct a $\sigma(n)$-stretch \ust\
for all graphs in polynomial time, we can obtain a polynomial-time
algorithm $\mathcal{A'}$ to construct an
$(O(\sigma^2),O(\sigma),\gamma)$-partition for all graphs and all
$\gamma > 0$ which uses $\mathcal{A}$ as a black box.}
\end{theorem}

\subsection{From partition hierarchies to a universal Steiner trees}
\label{sec:partn_to_ust}

\junk{Assume graph $G$ and root vertex $r$ to be fixed throughout this
  subsection. } We first prove the following important lemma,
showing the significance of $\mu$-respecting trees.  \junk{The
  main result here is an algorithm to construct an $O(\alpha^2 \beta^2
  \gamma \log n)$-stretch \ust\ from an
  $(\alpha,\beta,\gamma)$-partition hierarchy which is presented in
  Section~\ref{sec:top-down}. We say that a spanning tree $T$ of $G$
  {\em $\mu$-respects}\/ an $(\alpha, \beta, \gamma)$-partition
  hierarchy $\langle \partn_i \rangle$ if for any $i$, any cluster $C$
  of $\partn_i$, and any vertices $u, v \in C$, $\TreeDist{T}{u}{v}$
  is at most $\mu \alpha \gamma^i$.}

\begin{lemma}
\label{lem:stretch}
A spanning tree $T$ that $\mu$-respects an $(\alpha, \beta,
\gamma)$-partition hierarchy has a stretch of $O(\mu \alpha \beta
\gamma \log n)$.
\end{lemma}

\subsubsection{A basic bottom-up algorithm}
\label{sec:bottom-up}

We first present a bottom-up algorithm for constructing a spanning
tree $T$ from a partition hierarchy that strictly obeys it. Though the
stretch achieved by the spanning tree is much weaker than what we
obtain by a different algorithm, it helps develop our improved
algorithm.

\begin{algorithm}[ht!]
\caption{\ust: A basic bottom-up algorithm}
\label{alg:partn_to_ust_1}
\begin{algorithmic}[1]
  \REQUIRE Undirected graph $G$, $(\alpha, \beta, \gamma)$-partition
  hierarchy $\langle \partn_i: -1 \le i \le d = \lceil
  \log_\gamma(\frac{\Diam{G}}{\alpha}) \rceil \rangle$ for $G$.

  \ENSURE A spanning tree $T$ of $G$.

  \STATE For every cluster $C$ in $\partn_{-1}$, set $T(C)$ to $\emptyset$.

  \FOR{level $i$ from $0$ to $d$}

  \FOR{cluster $C$ in $\partn_i$}

  \STATE For an edge $e = (C_1, C_2)$ in
  $\superG{\subG{G}{C}}{\subP{\partn_{i-1}}{C}}$, let $m(e)$ denote
  the edge between $C_1$ to $C_2$ in $\subG{G}{C}$ that has minimum
  weight. (Recall that $\subG{G}{C}$ is the subgraph of $G$ induced
  by $C$ and $\subP{\partn_{i-1}}{C}$ is the partition $\partn_{i-1}$
  restricted to the set $C$.)
  
  \STATE Compute a shortest path tree $T'$ from an arbitrary source
  vertex in $\superG{\subG{G}{C}}{\subP{\partn_{i-1}}{C}}$.

  \STATE Set $T(C)$ to be the union of $\cup_{C' \in \subP{\partn_{i-1}}{C}} T(C')$ and
  $\{m(e) : e \in T'\}$.
  
  \ENDFOR

  \ENDFOR

  \STATE Set $T$ to be $T(V)$. (Note that $V$ is the lone cluster in
  $\partn_d$.)
\end{algorithmic}
\end{algorithm}

\begin{theorem}
\label{thm:universal}
For any graph $G$, given an $(\alpha, \beta, \gamma)$-partition
hierarchy, an $O((\alpha\beta)^{\log_\gamma n} \gamma \beta^2 \log
n)$-stretch \ust\ is computed by Algorithm~\ref{alg:partn_to_ust_1} in
polynomial time.
\end{theorem}

We complement the above construction by an almost matching lower bound
for stretch achievable by any spanning tree that strictly obeys a
partition hierarchy.

\begin{theorem}
\label{thm:lower_bound}
Let $\alpha$, $\beta < \gamma$. There exists a graph $G$ and an
$(\alpha,\beta,\gamma)$-partition hierarchy $\hier$ of $G$ such that
any spanning tree $T$ of $G$ that strictly obeys $\hier$ has sretch
$\Omega((\alpha\beta)^{\frac{\log_\gamma n}{4}} \gamma)$. (Depicted in
Figure~\ref{fig:lower_bound}.)
\end{theorem}

\subsubsection{Split and join: An improved top-down algorithm}
\label{sec:top-down}

The tree returned by Algorithm~\ref{alg:partn_to_ust_1} strictly obeys
the given partition hierarchy.  In doing so, however, it pays a huge
cost in the distances within the cluster which is unavoidable.

We now present a much more careful construction of a universal Steiner
tree which does not enforce the connectivity constraint within
clusters; that is, we use a given partition hierarchy $\hier$ to build
a tree $T$ in which $\subG{T}{C}$ may be disconnected.  By allowing
this disconnectivity within clusters, however, we show that we can
build a tree that $\mu$-respects the given hierarchy for a much
smaller $\mu$, assuming $\gamma$ is sufficiently large.  The
pseudocode is given in Algorithm~\ref{alg:ust}.  We have presented the
algorithm in a more general context where the goal is to compute a
forest rooted at a given set of portals.  To obtain a \ust, we invoke
the algorithm with the portal set being the singleton set consisting
of the root.

\junk{A
  critical step in our algorithm is to connect a given set of clusters
  using shortest paths, a problem that may be of independent
  interest.}

\junk{We first present some useful definitions.  For a directed path $v_0
\rightarrow v_1 \rightarrow \cdots \rightarrow v_k$, we call $v_0$ and
$v_1$ as the {\em tail}\/ and {\em head}\/, respectively, of the path.
We say that a directed path $p$ {\em traverses}\/ a set $C$ of
vertices if there exists a vertex $v$ in $C$ that lies in $p$ and is
not its head.}

\newcommand{\highway}[1]{\mbox{\sc highway}(#1)}
\newcommand{\portals}[1]{S_{#1}}
\newcommand{\favorite}{\mbox{\sc fav}}
\newcommand{\maxchild}{\mbox{\sc maxc}}
\newcommand{\rank}{\mbox{\sc rank}}

\begin{algorithm}[ht!]
\caption{\ust: The split and join algorithm}
\label{alg:ust}
\begin{algorithmic}[1]
  \REQUIRE Undirected graph $G = (V,E)$, a nonempty set $\portals{G} \subseteq
  V$ of portals, a partition hierarchy $\hier = \{\partn_0, \partn_1,
  \ldots, \partn_\ell\}$.

  \ENSURE A forest $F$ that connects every vertex in $V$ to
  $\portals{G}$.

  \STATE If the graph consists of a single vertex, then simply return
  the vertex as the forest.

  \STATE \junk{Let $\widehat{G}$ denote a graph with vertex set
  $\partn_\ell$ and an edge between cluster $C_1$ and $C_2$ in
  $\partn_\ell$ whenever there is an edge between $C_1$ and $C_2$ in
  $G$. } For an edge $e = (C_1, C_2)$ in $\superG{G}{\partn_\ell}$, let
  $m(e)$ denote the minimum-weight edge from $C_1$ to $C_2$ in $G$.

  \STATE Let $\widehat{S}$ denote the set of clusters that have a
  nonempty intersection with $\portals{G}$.  

  \STATE For every cluster $C$ in $\widehat{S}$, set $\portals{C}$ to be $C \cap S$.

  \STATE Compute a shortest path forest $\widehat{F}$ in
  $\superG{G}{\partn_\ell}$ rooted at $\widehat{S}$.

  \FOR{cluster $C$ in $\partn_\ell$ in order of decreasing distance
    from $\widehat{S}$ in $\widehat{F}$}

  \IF{ $C$ is a leaf node in $\widehat{F}$}

  \STATE Set $\rank(C)$ to be $0$, $\portals{C}$ to be $\{\mbox{tail
    of } m(e)\}$ where $e$ is the edge connecting $C$ to its parent in
    $\widehat{F}$.

  \ELSE 

  \STATE Let $\maxchild$ be $\max\{\rank(C') \mid C' \mbox{ is child
    of } C\}$.  Set $\favorite(C)$ to be a child of $C$ with rank
  $\maxchild$.  Set $\highway{C}$ to be a shortest path in $C$ from
  the head of $m(e)$ to the tail of $m(e')$ where $e$ and $e'$ are the
  edges connecting $\favorite(C)$ to $C$ and $C$ to its parent,
  respectively, in $\widehat{F}$.  Set $\portals{C}$ to be the set of
  nodes in $\highway{C}$.

  \IF{there exist at least two children of $C$ in $\widehat{F}$ whose rank equals $\maxchild$}

  \STATE Set $\rank(C)$ to be $\maxchild + 1$

  \ELSE 
  
  \STATE Set $\rank(C)$ to be $\maxchild$

  \ENDIF

  \ENDIF

  \ENDFOR

  \FOR{each cluster $C$ in $\partn_\ell$}

  \STATE \label{line:ust_call} Compute $F(C) = \ust(\subG{G}{C}, \portals{C}, \subH{\hier}{C})$

  \ENDFOR

  \STATE Return $F$ to be the union of $\bigcup_{C \in \partn_\ell}
  \highway{C}$, $\bigcup_{C \in \partn_\ell} F(C)$, and $\{m(e): e \in \widehat{F}\}$.
\end{algorithmic}
\end{algorithm}

\junk{In the following, we let $S_C$ denote the set of portals in the call
to $\ust$ on cluster $C$ in line~\ref{line:ust_call} of the algorithm.
We extend the definition of $S_C$ to the case where $C$ is the entire
graph in which case $S_C$ is simply the input $S$ to the algorithm.
}

\begin{lemma}
\label{lem:dist_in_ust}
The forest $F$ returned by the algorithm $7\alpha \beta$-respects $\hier$.
\end{lemma}

\begin{theorem}
\label{thm:partn_to_ust}
Given an undirected graph $G$, portal set $\portals{G} = \{r\}$, where
$r$ is an arbitrary vertex of $G$, and $(\alpha, \beta,
\gamma)$-partition $\hier$ of $G$ as input, Algorithm~\ref{alg:ust}
returns an $O(\alpha^2 \beta^2 \gamma \log n)$-stretch \ust.
\end{theorem}


\section{Partition hierarchy for general graphs}
\label{sec:hierarchical}
In this section we present our algorithm for obtaining a partition
hierarchy for general graphs. Our main result is the following.

\begin{theorem}
\label{thm:hierarchical_partition}
Fix integer $k \ge 1$ and $\epsilon > 0$. For any graph $G$, a
hierarchical $((\frac{4}{3} + \epsilon)4^{k-1} - \frac{4}{3},
kn^{\frac{1}{k}}, \gamma)$-partition can be constructed in polynomial
time for $\gamma \geq \frac{1}{\epsilon}((\frac{4}{3} +
\epsilon)4^{k-1} -\frac{4}{3})$. In paricular, setting $k = \lceil
\sqrt{\log n} \rceil$ and $\epsilon = 1$, a hierarchical
$(2^{O(\sqrt{\log n})}, 2^{O(\sqrt{\log n})}, 2^{O(\sqrt{\log
    n})})$-partition for any graph can be constructed in polynomial
time.
\end{theorem}

\vspace{-0.5cm}
\paragraph{Algorithm.} For notational convenience, as mentioned in
section~\ref{sec:defn}, we start buiiding the hierarchy at level $-1$
by defining $\mathcal{P}_{-1}$ as the trivial partition where every
vertex is in its own cluster. For $i = 0, 1, \ldots, d = \lceil
\log_{\gamma} \frac{\Diam{G}}{\alpha} \rceil$, we build the $i$th
level of the hierarchy, $\partn_i$, after building the previous
levels. Assuming that $\partn_{i-1}$ has been constructed, we
construct $\partn_i$, as follows.

\noindent {\bf Construction of level $i$:} Clusters of $\partn_i$ are formed in
successive stages starting from stage $0$. We assign a {\em rank} to
each cluster based on the stage in which it is created: a cluster
formed in stage $j$ gets the rank $j$. (All the clusters of level $-1$
are assigned the rank $0$.) We will denote the set of clusters of rank
$j$ in $\partn_i$ by $S^i_j$. At all times while building $\partn_i$,
we maintain a partitioning of the graph, i.e., we guarantee that each
vertex of the graph is contained in exactly one cluster of
$\partn_i$. The partitioning, however, may change as clusters of a
higher ranks are formed by merging clusters of lower ranks.

\noindent {\bf Stage $0$:} In stage $0$, we simply add all the clusters of
$\partn_{i-1}$ to $S^i_0$. 

\noindent {\bf Stage $j > 0$:} For $j > 0$, stage $j$ works in two
phases, one after another. 

\begin{mylowitemize}
\item {\bf First phase:} In the first phase, we repeatedly look for a
  vertex {\em contained in a cluster of rank at most $j-1$} such that
  the ball of radius $\gamma^i$ around it, $B(v,\gamma^i)$, intersects
  more than $n^{\frac{1}{k}}$ clusters of rank  precisely $j-1$. If we
  find such a vertex $v$, we merge the cluster containing $v$ with all
  the clusters of rank $j-1$ that $B(v,\gamma^i)$ intersects. This
  newly created cluster is assigned the rank $j$ and added to $S^i_j$
  while all the clusters that were merged to form it are deleted from
  their respective $S^i_{j'}$'s. The first phase ends when we can no
  longer find any such vertex $v$.

\item {\bf Second phase:} In the second phase, we repeat a simlar procedure
  for vertices {\em contained in clusters of rank $j$}. As long as we
  can find a vertex $v$ in a cluster of rank $j$ such that
  $B(v,\gamma^i)$ intersects more than $n^\frac{1}{k}$ clusters of
  rank $j-1$, we merge the cluster containing $v$ with all the
  clusters of rank $j-1$ that $B(v,\gamma^i)$ intersects to form a new
  cluster of rank $j$. We include this new cluster in $S^i_j$ and
  delete all the clusters that were merged to form it from their
  respective $S^i_{j'}$'s. The second phase, and also the stage $j$,
  ends when we cannot find any such vertex $v$, and the next stage
  begins.
\end{mylowitemize}

If no new cluster gets formed in the first phase of a stage, the
construction of level $i$ of the hierarchy finishes and
$\partn_i$ is defined as simply the union of all the non empty
$S^i_j$'s.
 
\smallskip
\BfPara{Remark} Although the two phases of a stage are quite similar
and one might be tempted to do away with this particular ordering of
mergings, the naive approach without the ordering does not
work. Having a careful order in which mergings are carried out enables
us to control the growth of the strong diameter of the clusters. To
see this, consider a cluster formed in the second phase of some stage
$j$. It contains a unique cluster that was formed in the first phase
of stage $j$. Call it the core. Our ordering ensures that only the
vertices in the core can lead to mergings in the second phase of stage
$j$. This is because for any vertex $v$ outside the core,
$B(v,\gamma^i)$ intersects at most $n^{\frac{1}{k}}$ clusters of rank
$j-1$, otherwise the first phase would not have ended. Thus the
mergings of the second phase cannot increase the diameter by too much
as the new vertices are always ``close'' to the core.

Section~\ref{proofs:hierarchical} contains the full pseudocode of the
algorithm, given in Algorithm~\ref{alg:hierarchical_partition}, and
the proof of Theorem~\ref{thm:hierarchical_partition}.  Using
Theorem~\ref{thm:hierarchical_partition} and
Theorem~\ref{thm:partn_to_ust}, we get our \ust construction for
general graphs.

\begin{corollary}
\label{cor:ust_general}
A $2^{O(\sqrt{\log n})}$-stretch universal Steiner tree can be computed in
polynomial time for any undirected graph.
\end{corollary}

\section{The Cluster aggregation problem}
\label{sec:cluster}
In this section, we define the Cluster Aggregation problem which
arises when building partition hierarchies for minor-free graphs (see
Section~\ref{section:minor-free}).  Our problem formulation and
algorithm, however, apply to arbitrary graphs and may be of
independent interest.  Indeed, our cluster aggregation algorithm is
useful for building other strong-diameter based hierarchical
partitions with applications to distributed
computing~\cite{birk+klsw:veracity}.

\begin{definition}[Cluster Aggregation]
\label{def:ccc}
Given a graph $G = (V,E)$, partition $\partn$ of $G$, set $S \subseteq
V$ of {\em portals}, a {\em cluster aggregation}\/ is a function
$\cntr: \partn \rightarrow S$.  The function $\cntr$ naturally induces
a new partition $\newpartn = \{ \bigcup_{C: \cntr(C) = s} C \mid s \in
S\}$ that coarsens $\partn$.  For each vertex $v$ in $V$, we define
the {\em detour}\/ $\detour{\cntr}{v}$ for $v$ under $\cntr$ to be the
difference between the distance from $v$ to $S$ in $G$ and the
distance from $v$ to $\cntr(\partn(v))$ in subgraph of $G$ induced by
the cluster in $\newpartn$ that contains $v$; i.e., $\detour{\cntr}{v}
= (\TreeDist{\subG{G}{\newpartn{v}}}{v}{\cntr(C)} - \dist{v}{S})$.  We
  define the detour of $\cntr$ to be $\max_{v \in V}
  \detour{\newpartn}{v}$.  The goal of the cluster merging problem is
  to find a cluster aggregation with minimum detour.
\end{definition}

Our algorithm for the Cluster Aggregation problem proceeds in $O(\log
n)$ phases.  Each phase has a number of iterations.  Each iteration
aggregates a subset of the clusters in $\partn$ and assigns the same
$\cntr$ value for each of them.  The selection of clusters in a
particular iteration is based on how shortest paths from these
clusters to $S$ proceed through the graph.  The interaction of these
shortest paths is captured by means of auxiliary directed graph.  For
any directed graph $K$ and set $A$ of vertices in $K$, let
$\inE{K}{A}$ (resp., $\outE{K}{A}$) denote the set of vertices that
have an edge into (resp., from) any vertex in $A$.  The pseudocode for
our algorithm appears in Algorithm~\ref{alg:merger} in
Section~\ref{proofs:cluster}.  We now show that
Algorithm~\ref{alg:merger} solves the Cluster Aggregation problem for
a given partition $\partn$ with a detour of
$O(\log^2(|\partn|)\MaxDiam(\partn))$.  We first establish the
following simple lemma that bounds the number of phases.  The proof of
the lemma and Theorem~\ref{thm:max-detour} are given in
Section~\ref{proofs:cluster}.

\begin{lemma}
\label{lem:phase}
If $V_i$ and $V_{i+1}$ are the set of vertices in $D_i$ and $D_{i+1}$
at the start of phase $i$ and $i+1$, respectively, then $|V_{i+1}| \le
|V_i|/2$.
\end{lemma}

\begin{theorem}
\label{thm:max-detour}
The detour for any vertex $v$ in $G$ in the cluster merger returned by
Algorithm~\ref{alg:merger} is at most $\log^2 (|\partn|)
\MaxDiam(\partn)$.
\end{theorem}

\newcommand{\cA}{{\cal A}}
\newcommand{\cB}{{\cal B}}
\newcommand{\cF}{{\cal F}}
\newcommand{\cI}{{\cal I}}
\newcommand{\cK}{{\cal K}}
\newcommand{\cN}{{\cal N}}
\newcommand{\cR}{{\cal R}}
\newcommand{\cZ}{{\cal Z}}
\newcommand{\wI}{{\widehat {\cal I}}}
\newcommand{\mydist}{d}

\section{Partition hierarchy for minor-free graphs}
\label{section:minor-free}
A weighted graph $G$ is $H$-minor free if zero or more edge
contractions on $G$ does not give a graph isomorphic $H$.  Minor-free
graphs are special cases of $k$-path separable graphs.  A graph $G$ is
{\em $k$-path separable} \cite{abraham+g:separators} if there exists a
subgraph $S$, called the {\em $k$-path separator}, such that: (i) $S =
S_{1} \cup S_{2} \cup \cdots \cup S_{l}$, where for each $1 \leq i
\leq l$, subgraph $S_i$ is the union of $k_i$ paths where each path is
shortest in $G \setminus \bigcup_{1 \leq j < i} S_j$ with respect to
its end points; (ii) $\sum_{i} k_i \leq k$; (iii) either $G \setminus
S$ is empty, or each connected component of $G \setminus S$ is
$k$-path separable and has at most $n/2$ nodes.

Thorup~\cite{thorup:planar} shows that any planar graph $G$ is
$3$-path separable, where all paths in the separator $S$ belong in
$S_1$, that is, they are shortest paths in $G$.  Abraham and
Gavoille~\cite{abraham+g:separators} generalize the result to any
$H$-minor free graph, where they show for fixed size $H$ the graph is $k$-path separable, for
some $k = k(H)$, and the $k$-path separator can be computed in
polynomial time.  Interesting classes of $H$-minor free graphs are:
planar graphs, which exclude $K_5$ and $K_{3,3}$; outerplanar graphs,
which exclude $K_4$ and $K_{2,3}$; series-parallel graphs, which
exclude $K_4$; and trees, which exclude $K_3$.  

\subsection{The algorithm}

 Consider now an arbitrary weighted $H$-minor free graph $G$, for
 fixed size $H$.  (You may also take $G$ to be an arbitrary $k$-path
 separable graph.)  \junk{We will construct a hierarchical $(\alpha,
   \beta, \gamma)$-partition of $G$ which is based on forming clusters
   around the path separators of $G$.  The concept of creating
   clusters around path separators has been introduced by Busch {\em
     et al.}  \cite{busch+lt:covers} in the context of sparse covers
   in $k$-path separable graphs.  Here, we extend that technique to
   hierarchical partitions.}  We build the partition hierarchy bottom
 up by coarsening clusters.  Suppose we are given a
 $(\alpha,\beta,\gamma^{i-1})$-partition $\partn_{i-1}$, where $i > 0$. 
We describe
 how to build a $(\alpha,\beta,\gamma^{i})$-partition $\partn_{i}$,
 such that $\partn_{i-1}$ is a refinement of $\partn_{i}$.  (Assume
 level $0$ partition has each node as a cluster.)

\smallskip
\BfPara{High-level recursive structure} The first clusters of
partition $\partn_{i}$ are formed around a $k$-path separator of $G$
by appropriately merging clusters of $\partn_{i-1}$ close to the
separator paths.  We then remove the $k$-path separator.  This may
result in the formation of one or more disjoint connected components,
each of which is still a $H$-minor free graph.  We repeat the
clustering process recursively on each residual connected component,
until no nodes remain.

\smallskip
\BfPara{Clustering a connected component} Algorithm
\ref{alg:minor-separator}, whose pseudocode is in
Section~\ref{proofs:minor-free}, implements the recursive
decomposition of $G$.  The algorithm actually receives as input an
arbitrary connected component $\Phi$ of $G$, which it then decomposes
into possibly one or more connected components that are processed
recursively.  The initial invocation is with $\Phi = G$.  The
resulting clusters of the $\partn_{i}$ partition of $G$ will appear in
a set $\cN$, which is initially empty.  New clusters formed around
path separators are inserted and maintained into $\cN$.  Some of the
newly formed clusters may merge with existing clusters in $\cN$
created form previously processed paths in earlier steps of the
algorithm.  The partition $\partn_{i}$ is the final $\cN$ that we
obtain after we recursively process all the path separators in each
component in $G$.  Let $S = S_{1} \cup S_{2} \cup \cdots \cup S_{l}$
be the path separator of $\Phi$.  We process the paths of $S$ in
sequence starting from the paths in $S_{1}$, then the paths in
$S_{2}$, and so on.

\smallskip
\BfPara{Processing a path} This is the main subroutine of the
algorithm, separated out as Algorithm~\ref{alg:path-clustering} in
Section~\ref{proofs:minor-free}.  Consider a path $p \in S_\chi$,
where $S_\chi$ is path set of $S$ in $\Phi$.  Let $\Psi$ be the
connected component of $\Phi \setminus \bigcup_{1 \leq j < \chi} S_j$
in which $p$ resides.  Algorithm \ref{alg:path-clustering} merges
clusters of $\partn_{i-1}$ which are within distance $2 \gamma^i$ from
$p$ using Algorithm \ref{alg:merger}.  As we show in the analysis, the
choice of this particular distance is to control the diameter of the
new clusters and the amount of intersections in any ball of diameter
$\gamma^i$.  The algorithm merges only integral clusters of
$\partn_{i-1}$ which are completely within $\Psi$, and which we denote
$\partn^\Psi_{i-1}$.  In Section~\ref{proofs:minor-free}, we show that
non-integral clusters have already been included in $\cN$ from
previously processed separator paths.  Let $\cA \subseteq
\partn^\Psi_{i-1}$ be the clusters within distance $2\gamma^i$ from
$p$ and which are candidates for merging.  We do not include in $\cA$
any cluster of $\partn^\Psi_{i-1}$ which has already been used in
$\cN$.  Of particular interest are the clusters $\cB \subseteq \cA$
which are neighbours with $\cN$ or next to non-integral clusters of
$\partn_{i-1}$, and these will be handled as special cases.

\junk{
Denote $\partn^\Psi_{i-1} \subseteq \partn_{i-1}$ the {\em
  integral} clusters of $\partn_{i-1}$ which are completely contained
in $\Psi$.  We define the following subsets $\cA$ and $\cB$ of
$\partn^\Psi_{i-1}$ such that: $\cA$ contains all clusters of
$\partn^\Psi_{i-1}$ not yet included in $\cN$ within distance $2
\gamma^i$ from $p$ in $\Psi$; $\cB$ contains all the clusters of $\cA$
which are adjacent to clusters in $\cN$ (where $\cN$ are the clusters
which have been formed so far form paths processed before $p$).

Let $\Psi'$ be the sub-graph induced by $\cA$ (note that $\Psi'$ may
not be connected).  
We define two sets of nodes $L$ and $U$ in $\Psi'$ which will serve as
portals around which new merged clusters will be formed.  Set $L$
contains the {\em leaders} of path $p$, which is a maximal set of
nodes in $p \cap \Psi'$, such for any pair $u,v \in L$,
$\mydist_p(u,v) \geq \gamma^i$, and $u$ and $v$ cannot belong to the
same cluster of $\cA$.  Set $U$ contains one arbitrary node from each
cluster in $\cB$ (for each cluster in $\cB$ that does not already have a leader).  
The clusters in $\cA$ are combined by invoking 
Algorithm \ref{alg:merger}
 portal nodes in $L \cup U$.

Let $\cR$ contain all resulting clusters from invoking Algorithm
\ref{alg:merger}.  We can write $\cR = \cI_p \cup \cK_p$ where $\cI_p$
consists of clusters that contain a node of $L$, and $\cK_p$ consists
of clusters that contain a node of $U$.  Each cluster $X \in \cK_p
\setminus \cI_p$ merges further with at most one arbitrary adjacent
cluster $Y \in \cN$, for which there is an edge $(u,v) \in E(\Psi)$
such that $u \in X$, $v \in Y$, and $v \notin \Psi'$.  We insert the
merged cluster from $X$ and $Y$ back to $\cN$.  The returned set of
clusters from processing path $p$ is $\cN = \cN \cup \cI_p$.  This
completes the processing of path $p$.

The full algorithm, formally described in
Algorithm~\ref{alg:minor-separator}, is initially invoked with $\Phi =
G$ and $\cN = \emptyset$.  The resulting partition $\partn_{i}$ is the
final $\cN$ that we obtain after we recursively process all the path
separators in $G$.
}

Let $\Psi'$ be the sub-graph induced by $\cA$ 
(note that $\Psi'$ may not be connected).
The clusters in $\cA$ are merged 
by invoking 
Algorithm \ref{alg:merger}.
We define two sets of nodes $L$ and $U$ in $\Psi'$ 
which serve as destination portals for the merged clusters.
Set $L$ contains the {\em leaders} of path $p$, 
which is a maximal set of nodes in $p \cap \Psi'$
such for any pair $u,v \in L$,
$\mydist_p(u,v) \geq \gamma^i$,
and $u$ and $v$ cannot belong to the same cluster of $\cA$.
Set $U$ contains one arbitrary node from each cluster in $\cB$
(from each cluster in $\cB$ that does not already contain a node in $L$).

Let $\cR$ contain all resulting merged (coarsen) clusters 
from invoking Algorithm \ref{alg:merger}.
We can write $\cR = \cI_p \cup \cK_p$ where
$\cI_p$
consists of clusters that contain a portal node of $L$,
and $\cK_p$
consists of clusters that contain a portal node of $U$.
Each cluster $X \in \cK_p$
may further merge with at most one arbitrary adjacent cluster $Y \in \cN$,
for which there is an edge $(u,v) \in E(\Psi)$
such that $u \in X$, $v \in Y$, and $v \notin \Psi'$.  
We insert the merged cluster from $X$ and $Y$ back to $\cN$. 
The returned set of clusters from processing path $p$
is $\cN = \cN \cup \cI_p \cup \cK'_p$,
where $\cK'_p$ contains the remaining clusters of $\cK_p$
(we actually show that $\cK'_p = \emptyset$).
 
\subsection{The analysis}
Consider a minor-free graph $G$ with $n$ nodes.
The recursive process of removing path separators
defines a decomposition tree $T$ of $G$.
Each node $t \in T$ corresponds to a connected component of $G$, 
which we will denote $G(t)$.
The root $\pi$ of $T$ corresponds to $G$, namely, $G(\pi) = G$. 
Denote $S(t)$ the path separator 
for the respective graph $G(t)$.
If $G(t) \setminus S(t) = \emptyset$, then $t$ is a leaf of $T$.
Otherwise, for each connected component $\Phi \in G(t) \setminus S(t)$
there is a node $w \in T$
such that $w$ is a child of $t$ and $G(w) = \Phi$.

According to the algorithm, after a new cluster is created (when a
path is processed) it may get larger when new clusters merge into it
(when subsequent paths are processed). 
Consider a path $p \in S(t)$, for some $t \in T$.  We say that a {\em
  cluster belongs} to $p$ if it contains a leader of $p$.  It can be
shown that a cluster in $\partn_i$ belongs to exactly one path.  A key
point of the analysis is that clusters of a path $p$ are far from
clusters in sibling nodes of $T$ (at least $2\gamma^i$ apart).  Thus,
when we bound intersections in balls of radius $\gamma^i$, we only
need to consider clusters on the same branch from the root to a leaf
of $T$.  Hence, the amount of intersections can be bounded using the
depth of $T$ which is $O(\log n)$, and the number of paths $k$ in a
separator.  Similarly, clusters can only grow along such a branch,
which also helps to control the diameter.  The proofs of the following
statements are deferred to Section~\ref{proofs:minor-free}.

\begin{lemma}
\label{thm:minor-partn_i}
$\partn_{i}$ is a $(\alpha', c_2 \alpha' k \log n , \gamma^i)$-partition,
where $\alpha' =  c_1 k \log^3 n$, for constants $c_1$ and $c_2$.
\end{lemma}

\begin{theorem}
\label{thm:minor-universal}
We can obtain a hierarchical 
$(O(\log^3 n), O(\log^4 n), O(\log^3 n))$-partition 
for any minor-free graph $G$ in polynomial time. 
\end{theorem}

From Theorems \ref{thm:partn_to_ust} and \ref{thm:minor-universal} we obtain the following corollary.

\begin{corollary}
\label{cor:ust_minor-free}
A $\polylog(n)$-stretch universal Steiner tree can be computed in
polynomial time for any minor-free graph with $n$ nodes.
\end{corollary}

\section{Proofs for Section~\ref{sec:ust_partn}}
\label{proofs:ust_partn}

\subsection{Proofs for Section~\ref{sec:ust_to_partn}}
\label{proofs:ust_to_partn}

\begin{LabeledProof}{Theorem~\ref{thm:ust_2_partn}}
Let $\gamma > 0$ be given. Assume we have algorithm $\mathcal{A}$ for
constructing $\sigma(n)$-stretch \ust\ in polynomial time. We
construct the algorithm $\mathcal{A}'$ that constructs an
$(O(\sigma(n)^2),O(\sigma(n)),\gamma)$-partition as follows. For a
given $n$, let $\sigma$ denote $\sigma(n)$. Given an $n$-node weighted
undirected graph $G=(V,E,w)$, $\mathcal{A}'$ constructs graph
$G'=(V',E',w')$ where $V' = V \cup \{r\}$, $E' = E \cup \{(r,v): v \in
V\}$ and $w'$ extends $w$ to $E'$ by simply assigning $w((r,v)) = 2
\sigma \gamma$ for all $v \in V$. Here $r$ is an additional vertex
not in $V$. $\mathcal{A}'$ invokes $\mathcal{A}$ with graph $G'$ and
root vertex $r$ as inputs. Let $T$ be the tree rooted at $r$ output by
$\mathcal{A}$ and $T_1, \ldots, T_k$ be the subtrees of $T$ connected
directly to the root $r$ by single edges.  $\mathcal{A}'$ simply
outputs the partition $\partn = \{C_1, \ldots, C_k\}$, where $C_i$ is
the set of vertices in $T_i$.  We now argue that $\partn$ is a
$(O(\sigma^2),O(\sigma),\gamma)$-partition of $G$.

\begin{lemma}
The strong diameter of each $C_i$ is at most $4 \sigma (\sigma - 1)
\gamma$.
\end{lemma}
\begin{proof}
Fix a $C_i$. It is enough for us to prove that the height of the tree
$T_i$ is at most $2 \sigma (\sigma - 1) \gamma$ as we can reach any
vertex in $C_i$ from any other while remaining within $C_i$ by going
through the root of $T_i$. Assume not. Then there is a vertex $v$ in
tree $T_i$ whose distance in this tree from the root of $T_i$ is more
than $2 \sigma (\sigma - 1) \gamma$. Consider the graph $G'$ with the
root vertex $r$ for which $\mathcal{A}$ returned
$T$. $\cost{T_{\{v\}}}$ is more than $2 \sigma \gamma + 2 \sigma
(\sigma - 1) \gamma = 2 \sigma^2 \gamma$, while $\Opt{\{v\}}$ is $2
\sigma \gamma$. Thus $\frac{\cost{T_{\{v\}}}}{\Opt{\{v\}}} > \sigma$,
which contradicts the fact that $T$ is a $\sigma$-stretch \ust\ for
$G'$.
\end{proof}

\begin{lemma}
For any vertex $v \in V$, $B(v, \gamma)$ intersects at most $2 \sigma$
clusters of $\partn$.
\end{lemma}
\begin{proof} 
The proof is by contradiction.  Suppose there is a vertex $v$ such
that $B(v, \gamma)$ intersects $d > 2 \sigma$ clusters of $\partn$.
We select one vertex from each of these $d$ different clusters such
that the selected vertices lie in $B(v, \gamma)$, and call this set
$S$. Now consider the graph $G'$ with the root vertex $r$ for which
$\mathcal{A}$ returned $T$. Since each vertex in $S$ lies in a
different $T_i$ in $T$, $\cost{T_{S}}$ is at least $2 \sigma \gamma
d$. On the other hand, $\Opt{S}$ is at most $2 \sigma \gamma + d
\gamma = (2 \sigma + d) \gamma$ as $v$ is at a distance $2 \sigma
\gamma$ from $r$ and each of the $d$ vertices in $S$ are at most a
distance $\gamma$ away from $v$.  Thus
$\frac{\cost{T_{\{v\}}}}{\Opt{\{v\}}} = \frac{2 \sigma d}{2 \sigma +
  d} > \sigma$ by our choice of $d$, which again contradicts the fact
that $T$ is a $\sigma$-stretch \ust\ for $G'$.
\end{proof}

The theorem follows from the above two lemmas.
\end{LabeledProof}

\subsection{Proofs for Section~\ref{sec:partn_to_ust}}
\label{proofs:partn_to_ust}

\begin{LabeledProof}{Lemma~\ref{lem:stretch}}
Let $\langle \partn_i\rangle$ denote the given $(\alpha, \beta,
\gamma)$-partition hierarchy.  Fix a non-empty set $X$ of
vertices. Note that $X$ is assumed to not contain the root $r$.  For
each cluster $C$ in the partition hierarchy such that $C \cap (X \cup
\{r\})$ is nonempty, let $v(C)$ denote an arbitrary vertex in $C
\cap (X \cup \{r\})$.

We place an upper bound on the cost of $T_X$, the subgraph of $T$
connecting the vertices in $X$ to the root $r$, as follows.  Let $n_i$
denote the number of clusters in $\partn_i$ that $X \cup \{r\}$
intersects. Since we have defined $\partn_{-1}$ to be the trivial
clustering consisting of a singleton set for each vertex, $n_{-1}$ is
simply $|X \cup \{r\}|$. Let $j$ be the smallest integer such that $X$
is a subset of the cluster in $\partn_j$ that contains $r$. In other
words, $n_j$ equals $1$ and $n_{i} > 1$ for all $-1 \le i < j$. Fix an
$i$, $-1 \le i < j$. Let $C$ be any cluster of $\partn_{i}$ that
intersects $X \cup \{r\}$, and let $C'$ denote the cluster of
$\partn_{i+1}$ that contains $C$.  Since $T$ $\mu$-respects the
partition hierarchy, it follows that the length of the path from
$v(C)$ to $v(C')$ in $T$ is at most $\mu \alpha \gamma^{i+1}$.
Therefore, the cost of $T_X$ is at most $\sum_{-1 \le i < j} n_{i} \mu
\alpha \gamma^{i+1}$. Let $I = \{i \mid (i = j) \vee (-1 \le i < j \wedge
\exists p: n_i \ge 2^p \wedge n_{i+1} < 2^p)\}$. For $\ell \in I$, let
$I_\ell = \{i \mid (-1 \le i \le \ell) \wedge \neg (\exists \ell' \in I: i
\le \ell' < \ell)\}$. We have,
\[
\sum_{i \in I_\ell} n_{i} \mu \alpha \gamma^{i+1} \leq \sum_{i \in
  I_\ell} 2 n_\ell \mu \alpha \gamma^{i+1} \leq \sum_{-1 \le i \le
  \ell} 2 n_\ell \mu \alpha \gamma^{i+1} = O(n_\ell \mu \alpha
\gamma^{\ell + 1})
\]

\junk{ Let $x$ be any in $C \cap X$, and let $P_x$ and $P_y$ denote
  the path from $x$ and $y$, respectively, to $r$ in $T$.  By our
  construction, the subpaths of $P_x$ and $P_y$ in $C'$ are identical,
  and of length at most $(\alpha \beta \gamma)^{i+1}$.  It thus
  follows that the cost of $T_X$ is at most
\[
\sum_{0 \le i \le j} n_{i-1} [(\alpha \beta\gamma)^{i+1} - 1].
\]
}

We next place a lower bound on $\Opt{X}$.  Fix an $i$, $0 \le i < j$.
By the cluster-valence property of the hierarchy, any ball of radius
$\gamma^i$ intersects at most $\beta$ clusters in $\partn_i$.  Thus,
there are at least $\lceil n_i/\beta \rceil$ vertices in $X$ that are
at pairwise distance at least $\gamma^i$ from one another.  This
implies that $\Opt{X}$ is at least $(\lceil n_i/\beta \rceil - 1)
\gamma^i$.  If $\lceil n_i/\beta \rceil = 1$, we invoke the padding
property which says there is at least one vertex in $X$ that is at
distance at least $\gamma_i$ from the root, implying a lower bound of
$\gamma^i$ on $\Opt{X}$.  Combining the two bounds, we obtain a lower
bound of $\Omega(n_i \gamma^i/\beta)$.  For $i=-1$, we also have a
lower bound of $n_{-1}$ since the minimum edge-weight is 1.  Noting
that $|I| = O(\log n)$, we get the stretch of $T(G)$ to be
\[
O\left(\sum_{\ell \in I} \frac{\sum_{i \in I_\ell} n_i \mu \alpha \gamma^{i+1}}{\Opt{X}}\right)
= O\left(\sum_{\ell \in I} \frac{n_\ell \mu \alpha \gamma^{\ell+1}}{n_\ell \gamma^\ell/\beta}\right)
=
O\left(\sum_{\ell \in I} \mu \alpha \gamma^{\ell+1} \beta/
\gamma^{\ell}\right) = O(\mu \alpha \beta \gamma \log n).
\]
\end{LabeledProof}

\subsubsection{Proofs for Section~\ref{sec:bottom-up}}
\label{proofs:bottom-up}

\begin{LabeledProof}{Theorem~\ref{thm:universal}}
Given a graph $G$ and a partition hierarchy $\langle \partn_i \rangle$
for $G$, the algorithm builds a spanning tree $T = T(V)$ by
iteratively building spanning trees for each cluster of the partition
hierarchy in a bottom-up manner. We first show by induction on $i \geq
0$ that for any cluster $C \in \partn_i$, the strong diameter of
$T(C)$ is at most $(\alpha \beta \gamma)^{i+1} - 1$.  The induction
basis directly follows from the strong diameter property of an
$(\alpha, \beta, 1)$-partition.

\junk{
For $i$ from $0$ to $\log_\gamma D$, we repeat the following:
\begin{itemize}
\item
For each cluster $C$ in $\partn_i$, we compute a spanning tree $T(C)$
of $C$ as follows.  Let $S$ be the set of disjoint clusters in
$\partn_{i-1}$ whose union equals $C$.  Let $G_S$ denote a graph with
vertex set $S$ and an edge between cluster $C_1$ and $C_2$ in $S$
whenever there is an edge between $C_1$ and $C_2$ in $G$.  For an edge
$e = (C_1, C_2)$ in $G_S$, let $m_S(e)$ denote an arbitrary edge from
$C_1$ to $C_2$ in $G$.  Compute a shortest path tree $T'$ from an
arbitrary source vertex in $G_S$.  We set $T(C)$ to be the union of
$\cup_{C' \in S} T(C')$ and $\{m_S(e) : e \in T'\}$.
\end{itemize}
}

We now establish the induction step.  Let $C$ be a cluster in
$\partn_i$ and let $u$ and $v$ be two vertices in $C$.  By the
hierarchy property, $C$ is the union of a set, say $S$, of clusters in
$\partn_{i-1}$.  Since $\partn_i$ is an $(\alpha, \beta,
\gamma^i)$-partition, it follows that the strong diameter of $C$ is at
most $\alpha \gamma^i$.  Hence, there exists a path $P$ between $u$
and $v$ in $C$ of length at most $\alpha \gamma^i$.  By the
intersection property, any ball of radius $\gamma^{i-1}$ intersects at
most $\beta$ clusters in $\partn_{i-1}$, and hence at most $\beta$
clusters in $S$.  Therefore, the path $P$ intersects at most $\alpha
\beta \gamma$ clusters in $S$.  Thus, the diameter of $G_S$ is at most
$\alpha \beta \gamma$.  By the induction hypothesis, it follows that
the strong diameter of $T(C)$ is at most $\alpha \beta \gamma - 1 +
\alpha\beta\gamma((\alpha\beta\gamma)^{i} - 1)$, which equals
$(\alpha\beta\gamma)^{i+1} - 1$.

Since $(\alpha\beta\gamma)^{i+1}/(\alpha \gamma^i)$ is at most
$\alpha^k \beta^{k+1}\gamma$ for $k = \log_\gamma n$, it follows that
$T$ $(\alpha^{k-1} \beta^{k+1} \gamma)$-respects the strong partition
hierarchy.  By Lemma~\ref{lem:stretch}, we obtain that $T$ has stretch
at most $O((\alpha \beta)^{\log_\gamma n} \gamma \beta^2 \log n)$,
completing the proof of the theorem.

\junk{
We now show that for appropriate choices of $\alpha$, $\beta$, and
$\gamma$, the stretch of $T(G)$ is at most $2^{\log^{3/4} n}$.

Fix a set $X$ of vertices, that contains the root $r$ and is of size
at least two.  Let $T_X$ denote the subtree of $T(G)$ induced by $X$.
We place an upper bound on the cost of $T_X$ as follows.  Let $n_i$
denote the number of clusters in $\partn_i$ that $X$ intersects.  For
convenience, we define $n_{-1}$ to be $|X|$ and $\partn_{-1}$ to be
the trivial clustering consisting of a singleton set for each vertex.
Let $j$ be the smallest integer such that $X$ is a subset of the
cluster in $\partn_j$ that contains $r$.  Thus $n_j$ equals $1$ and
$n_{i-1} > 1$ for all $0 \le i \le j$.  Fix an $i$, $0 \le i \le j$.
Let $C$ be any cluster of $\partn_{i-1}$ that intersects $X$, and let
$C'$ denote the cluster of $\partn_i$ containing $C$.  Let $x$ and $y$
be two different vertices in $C \cap X$, and let $P_x$ and $P_y$
denote the path from $x$ and $y$, respectively, to $r$ in $T$.  By our
construction, the subpaths of $P_x$ and $P_y$ in $C'$ are identical,
and of length at most $(\alpha \beta \gamma)^{i+1}$.  It thus follows
that the cost of $T_X$ is at most
\[
\sum_{0 \le i \le j} n_{i-1} [(\alpha \beta\gamma)^{i+1} - 1].
\]
We next place a lower bound on $OPT(X)$.  Fix an $i$, $0 \le i < j$.
By the intersection property, any ball of radius $\gamma^i$ intersects
at most $\beta$ clusters in $\partn_i$.  Thus, there are at least
$\lceil n_i/\beta \rceil$ vertices in $X$ that are at pairwise
distance at least $\gamma^i$ from one another.  This implies that
$OPT(X)$ is at least $(\lceil n_i/\beta \rceil - 1) \gamma^i$.  If
$\lceil n_i/\beta \rceil = 1$, we invoke the padding property which
says there is at least one vertex in $X$ that is at distance
$\gamma_i$ from the root, implying a lower bound of $\gamma^i$ on
$OPT(X)$.  Combining the two bounds, we obtain a lower bound of
$\Omega(n_i \gamma^i/\beta)$.  For $i=-1$, we also have a lower bound
of $n_{-1}$ (which in fact supersedes the lower bound for $i=0$).
Thus, the stretch of $T(G)$ is at most 
\[
O\left(\sum_{0 \le i \le j} (\alpha \beta \gamma)^{i+1} \beta/ \gamma^{i-1}\right) = O((\alpha\beta)^j\beta \gamma^2).
\]
Since $j$ is at most $\log_\gamma n$, we obtain an upper bound of
$(\alpha\beta)^{\log_\gamma n} \beta \gamma^2$ on the stretch of $T$,
completing the proof of the lemma.}
\end{LabeledProof}

\begin{figure}[ht]
\begin{center}
\includegraphics[width=5in]{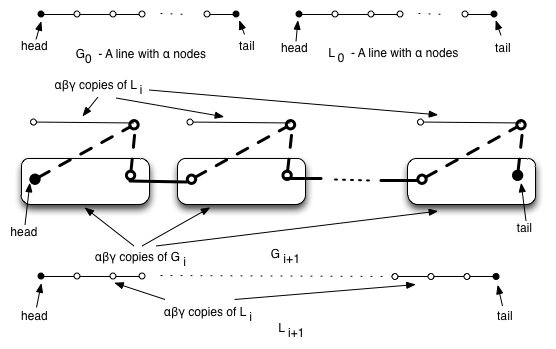}
\caption{Construction for the lower bound of
  Theorem~\ref{thm:lower_bound}. The figure shows the recursive
  construction of graphs $G_i$ and $L_i$. The bold solid edges have
  length $\frac{\gamma^{i+1}}{\beta}$ while the bold dashed edges have
  length $\gamma^{i-1}$. The vertices shown represent the head and the
  tail vertices.}
\label{fig:lower_bound}
\end{center}
\end{figure}

\begin{LabeledProof}{Theorem~\ref{thm:lower_bound}}
The construction is depicted in Figure~\ref{fig:lower_bound}. We
recursively define graphs $G_i$ and $L_i$. We have two distinguished
vertices, called {\em head} and {\em tail}, for each $G_i$ and
$L_i$. Both $G_0$ and $L_0$ are defined to be simply lines with
$\alpha$ vertices and all edges have length $1$. $L_{i+1}$ is defined
recursively by joining together $\alpha\beta\gamma$ copies of $L_i$,
by connecting the head of the successive copy to the tail of the
previous one by an edge of length $1$. The head and tail of $L_{i+1}$
is simply defined as the head of the first copy and the tail of the
last copy of $L_i$ respectively. $G_{i+1}$ is recursively defined by
putting together $\alpha\beta\gamma$ copies of $G_i$ and
$\alpha\beta\gamma$ copies of $L_i$ as shown in
Figure~\ref{fig:lower_bound}. The edges connecting the tail of one
copy of $G_i$ to the head of the successive copy of $G_i$ has length
$frac{\gamma^i}{\beta}$. The two edges connecting a $G_i$ to the
adjacent $L_i$ has length $2\gamma^{i-1}$. The head and tail of
$G_{i+1}$ is simply defined as the head of the first copy and the tail
of the last copy of $G_i$ respectively. 

We can see inductively that the number of vertices in $L_i$ is $\alpha
(\alpha\beta\gamma)^i$, and the number of vertices in $G_i$ is $\alpha
i (\alpha\beta\gamma)^i$. We define $i_0$ to be the solution for
$\alpha i (\alpha\beta\gamma)^i = n$ and define $G_{i_0}$ to be the
graph $G$. Define the head of $G_{i_0}$ to be the root vertex.

The hierarchical partition $\hier$ of $G$ is simply defined as
follows. The clusters at level $i$ are all the copies of $L_i$ and
$G_i$ in $G$. It is not hard to see that $\hier$ is a hierarchical
$(\alpha, \beta, \gamma)$-partition of $G$.

Let $T$ be any spanning tree of $G$ that strictly obeys the hierarchy
$\hier$. It is easy to see that the distance between the head and the
tail of any $G_i$ in the subtree induced by it is at least $\alpha
(\alpha\beta\gamma)^i$. On the other hand, there is a path of length
at most $2\gamma^{i-1}$ between them. We conclude that the stretch of
$T$ is at least $\alpha^{i_o}\beta^{i_0-1}\gamma$ which is
$\Omega((\alpha\beta)^{\frac{\log_\gamma n}{4}} \gamma)$.
\end{LabeledProof}

\subsubsection{Proofs for Section~\ref{sec:top-down}}
\label{proofs:top-down}

\begin{lemma}
\label{lem:tree}
The output $F$ of the algorithm is a spanning forest, each tree 
containing exactly one vertex in $\portals{G}$.
\end{lemma}
\begin{proof}
The proof is by induction on the number of recursive calls to the UST
algorithm.  For the induction base, we consider the case where the
graph consists of a single vertex; in this case, the algorithm returns
the vertex as the forest, which satisfies the desired claim.

For the induction step, we note that the forest $F$ returned is the
union of three sets: (a) union of $\highway{C}$ over all $C$ in
$\partn_\ell$; (b) union of $F(C)$ over all $C$ in $\partn_\ell$; and
(c) $\{m(e): e \in \widehat{F}\}$.  By the induction hypothesis, each
$F(C)$ is a forest spanning $C$, each tree of which contains exactly
one vertex of $\portals{C}$.  We distinguish between two kinds of
clusters.  If $C$ is in $\widehat{S}$, then $F(C)$ is a forest, each
tree of which contains exactly one vertex of $\portals{G} \cap C$.
Otherwise, $F(C)$ is a forest, each tree of which contains exactly one
vertex of $\highway{C}$.  It thus follows that the union of (a) and
(b) above gives a forest for each cluster $C$ satisfying the following
condition: if $C$ is in $\widehat{S}$, the forest contains a spanning
forest of $C$, each tree of which contains exactly one vertex of
$\portals{G} \cap C$; otherwise, the forest contains a spanning tree
of $C$.  

Finally, the edges of (c) connect the clusters not in $\widehat{S}$ to
the clusters in $\widehat{S}$ via a forest.  Consequently, adding
these edges to the forest formed by (a) and (b) yields a spanning
forest over $G$, each tree of which contains exactly one vertex in
$\portals{G}$.
\end{proof}

\begin{lemma}
\label{lem:highway_dist}
Let $F$ be the final forest returned by the algorithm.  For any
cluster $C$, when \ust\ is called on cluster $C$, either $\portals{C}$
is a subset of $\portals{G}$ or for any two vertices $u$ and $v$ in
$S_C$, $\TreeDist{F}{u}{v}$ is at most $\TreeDist{C}{u}{v}$.
\end{lemma}
\begin{proof}
We first prove that for any cluster $C$, the set $\portals{C}$ is
exactly one of the following: (i) $\portals{G}$, if $C$ is $G$; or
(ii) a subset of $\portals{C'}$ for the parent cluster $C'$ of $C$; or
(iii) a subset of nodes on a $\highway{C}$ constructed when processing
parent cluster $C'$.  The proof is by induction on the level of the
hierarchy.  The base case is trivial for $C$ being the whole graph.
For the induction step, consider level $i \le \ell$ of the hierarchy,
and let $C$ be a cluster in $\partn_i$.

We consider two cases.  In the first case, $C$ intersects
$\portals{C'}$ where $C'$ is the parent cluster for $C$.  In this
case, $\portals{C}$ is set to the intersection of $C$ and
$\portals{C'}$ as desired.  In the second case, $C$ is disjoint from
$\portals{C'}$.  In this case, $\portals{C}$ is simply the set of
vertices in $\highway{C}$, again completing the induction step.

To complete the proof of the lemma, consider a cluster $C$ in
$\partn_i$, for some $i$.  If $\portals{C}$ is a subset of
$\portals{G}$, the lemma trivially follows.  If $\portals{C}$ is not a
subset of $\portals{G}$, then by the above claim, $\portals{C}$ equals
the set of nodes in $\highway{C}$.  By our construction, $\highway{C}$
is a shortest path in $C$.  Since $\highway{C}$ is part of $F$, it
follows that for any two vertices $u$ and $v$ in $\portals{C}$,
$\TreeDist{F}{u}{v}$ equals $\TreeDist{C}{u}{v}$.
\end{proof}

\begin{lemma}
\label{lem:rank}
The rank of any cluster $C$ in partition $\partn_\ell$ is at most
$\log(|\partn_\ell|)$.
\end{lemma}

\begin{proof}
Let $\widehat{F}$ denote the shortest path forest in
$\superG{G}{\partn_\ell}$.  We show that for any cluster $C$, the rank
of $C$ is at most $\log(m_C)$, where $m_C$ is the number of nodes in
the subtree of $\widehat{F}$ rooted at $C$.

The proof is by induction on the height of $C$.  The induction basis
is immediate for the leaves of $\widehat{F}$.  We now consider the
induction step.  For cluster $C$, let $r$ denote the rank of the child
with highest rank among all children of $C$.  Let $Z$ denote the set
of children of $C$ that have rank $r$.  We note that $m_C$ is at least
$1 + \sum_{C' \in Z} m_{C'}$.  Furthermore, by the induction
hypothesis, $m_{C'}$ is at least $2^{r}$, for each $C'$ in $Z$.  We
consider two cases.  If $|Z|$ is $1$, then the rank of $C$ equals $r$,
which is at most $\log(m_C)$ by the induction hypothesis.  In $|Z| \ge
2$, then the rank of $C$ equals $r + 1$; since $m_C$ is at least $1 +
2\cdot 2^r > 2^{r+1}$, the induction step follows, completing the
proof.
\end{proof}

\begin{lemma}
\label{lem:dist_to_portal}
Let $F$ be the final forest returned by the algorithm.  If $\gamma \ge
3\log n$, then for any cluster $C$ in $\partn_i$ and vertex $u$ in $C$,
$\TreeDist{F}{u}{\portals{C}}$ is at most $3\alpha^2 \beta \gamma^i$.
\end{lemma}
\begin{proof}
We prove by induction on level $i$ that $\TreeDist{F}{u}{\portals{C}}$
is at most $3 \alpha^2 \beta \gamma^i$, with the base case being $i =
0$.  In this case, the cluster and its portal set are the same
singleton vertex set, trivially yielding the desired claim.  For the
induction step, we consider $i > 0$.  Let $C$ be a cluster of
$\partn_i$.  For any vertex $u$ in $C$, let $C_u$ denote
$\partn_{i-1}(u)$, that is, the cluster in partition $\partn_{i-1}$
that contains $u$.

As in the algorithm, let $\widehat{S}$ denote the set of clusters in
the partition of $C$ that intersect $\portals{C}$.  Let $C_u = C_0,
C_1, \ldots, C_k$, where $C_k \in \widehat{S}$, denote the sequence of
clusters in the unique path from $C_u$ to $\widehat{S}$ in
$\superG{G}{\partn_\ell}$, which we refer to as the supergraph in the
following argument.  Note that $C_i$ is the parent of $C_{i-1}$ in the
supergraph.  By our argument in the proof of
Theorem~\ref{thm:universal}, we know that $k$ is at most $\alpha \beta
\gamma$.  We now argue that there are at most $\log n$ elements $C_i$
in the sequence such that $C_i$ is not $\favorite(C_{i+1})$.  To see
this, we note that if $C_i$ is not $\favorite(C_{i+1})$, then
$\rank(C_{i+1})$ strictly exceeds $\rank(C_i)$.  Since the rank of any
cluster is at most $\log n$ by Lemma~\ref{lem:rank}, the desired claim
holds.

This sequence of clusters induces a path from $u$ to $\portals{C}$,
which consists of (a) the connecting edges in the supergraph, (b) the
highway in each cluster $C_i$ in the sequence, (c) for each cluster
$C_i$ such that $C_{i-1}$ is not a favorite of $C_{i}$, the unique
path in $F(C_i)$ (and, hence, in $F$) that connects the head of the
edge connecting $C_{i-1}$ and $C_i$ to $\portals{C_i}$.  Since the
number of clusters in the sequence is at most $\alpha \beta \gamma$,
and the highway in each cluster is a shortest path of length at most
$\alpha \gamma^{i-1}$, the total length of the paths in (a) and (b) is
at most $2\alpha^2 \beta \gamma^i$.  The number of clusters in (c) is
at most $\log n$, and by the induction hypothesis, the length of each
path in (c) is at most $3\alpha^2 \beta \gamma^{i-1}$.  We thus have,
\begin{eqnarray*}
\TreeDist{F}{u}{\portals{C}} & \le & 2\alpha^2 \beta \gamma^i + (3\log
n) \alpha^2 \beta \gamma^{i-1}\\ 
& \le & 3\alpha^2 \beta \gamma^i
\end{eqnarray*}
for $\gamma\ge 3 \log n$, thus completing the proof of the lemma.
\end{proof}

\begin{LabeledProof}{Lemma~\ref{lem:dist_in_ust}}
We show that for any cluster $C$ in $\partn_i$, and vertices $u, v$ in
$C$, $\TreeDist{F}{u}{v}$ is at most $7\alpha^2 \beta \gamma^i$; this
will establish the desired claim.  By Lemma~\ref{lem:dist_to_portal},
$\TreeDist{F}{u}{S_C}$ and $\TreeDist{F}{v}{S_C}$ are both at most
$3\alpha^2 \beta \gamma^i$.  By Lemma~\ref{lem:highway_dist}, for any
two nodes $x$ and $y$ in $S_C$, $\TreeDist{F}{x}{y}$ is at most the
strong diameter of $C$, which is at most $\alpha \gamma^i$.  Putting
these three distances together, we obtain that $\TreeDist{F}{u}{v}$ is
at most $7 \alpha^2 \beta \gamma^i$.
\end{LabeledProof}

\begin{LabeledProof}{Theorem~\ref{thm:partn_to_ust}}
By Lemma~\ref{lem:tree}, the output $F$ is a spanning forest, each
tree of which contains exactly one vertex of $\portals{G}$.  Since
$\portals{G}$ has only one vertex, the forest $F$ returned is a tree.
By Lemma~\ref{lem:dist_in_ust}, $F$ $(7 \alpha \beta)$-respects
$\hier$.  By Lemma~\ref{lem:stretch}, we obtain that $F$ has stretch
$O(\alpha^2 \beta^2 \gamma \log n)$.
\end{LabeledProof}

\section{Proofs for Section~\ref{sec:hierarchical}}
\label{proofs:hierarchical}
\begin{algorithm}[ht!]
\caption{Algorithm to obtain a partition hierarchy for general graphs}
\label{alg:hierarchical_partition}
\begin{algorithmic}[1]

\REQUIRE A weighted graph $G=(V,E,w)$, integer $k$, $\gamma \geq
\frac{1}{\epsilon}((\frac{4}{3}+\epsilon)4^{k-1} - \frac{4}{3})$ 

\ENSURE A hierarchical $(\alpha=(\frac{4}{3}+\epsilon)4^{k-1}
- \frac{4}{3},\beta=kn^{\frac{1}{k}},\gamma)$-partition of $G$

\STATE Define $\partn_{-1}$ to be the trivial partition where each
vertex of $V$ is in its own cluster, \\ i.e., $\partn_{-1} = \{\{v\}: v
\in V\}$.

\FOR{level $i$ from $0$ to $\lceil\log_\gamma(\frac{\Diam{G}}{\alpha})
  \rceil$}

\STATE $S^i_0 = \mathcal{P}_{i-1}$.
\STATE $S^i_j = \emptyset$ for all $1 \leq j \leq k-1$.
\STATE $j \leftarrow 1$.

\WHILE{$j < k$ and $S^i_{j-1} \neq \emptyset$}

\WHILE{there exists a $v$ such that $v \in C_v$ for some $C_v \in S^i_{j_v}$ and $j_v < j$, \\
and $B(v,\gamma^i)$ intersects more than $n^{\frac{1}{k}}$ clusters from $S^i_{j-1}$}
\STATE Delete $C_v$ from $S^i_{j_v}$, i.e., $S^i_{j_v} = S^i_{j_v} \setminus \{C_v\}$.
\STATE Delete all the clusters of $S^i_{j-1}$ that $B(v,\gamma^i)$ intersects from it, i.e., \\
$S^i_{j-1} = S^i_{j-1} \setminus \{C: C \in S^i_{j-1} \wedge B(v,\gamma^i) \cap C \neq \emptyset\}$.
\STATE Merge $C_v$ and all the clusters deleted from $S^i_{j-1}$ and add to $S^i_j$, i.e., \\
$S^i_j = S^i_j \cup C_v \cup \left(\bigcup_{C \in {\cal X}} C\right)$, where ${\cal X}$ equals $\{C \in S^i_{j-1}: B(v,\gamma^i) \cap C \neq \emptyset\}$.
\ENDWHILE

\WHILE{there exists a $v$ such that $v \in C_v$ for some $C_v \in S^i_j$, \\
and $B(v,\gamma^i)$ intersects more than $n^{\frac{1}{k}}$ clusters from $S^i_{j-1}$}
\STATE Delete $C_v$ from $S^i_j$, i.e., $S^i_j = S^i_j \setminus \{C_v\}$.
\STATE Delete all the clusters of $S^i_{j-1}$ that $B(v,\gamma^i)$ intersects from it, i.e., \\
$S^i_{j-1} = S^i_{j-1} \setminus \{C: C \in S^i_{j-1} \wedge B(v,\gamma^i) \cap C \neq \emptyset\}$.
\STATE Merge $C_v$ and all the clusters deleted from $S^i_{j-1}$ and add to $S^i_j$, i.e., \\
$S^i_j = S^i_j \cup C_v \cup \left(\bigcup_{C \in {\cal Y}} C\right)$, where ${\cal Y}$ equals $\{C \in S^i_{j-1}: B(v,\gamma^i) \cap C \neq \emptyset\}$.
\ENDWHILE

\STATE $j = j + 1$.

\ENDWHILE

\STATE $\mathcal{P}_i = \cup_{t = 0}^{t = k-1} S^i_t$.

\ENDFOR

\STATE Output $(\partn_0, \ldots, \partn_{\lceil\log_\gamma(\frac{\Diam{G}}{\alpha}) \rceil})$.

\end{algorithmic}
\end{algorithm}

We have the following claims that bound the size and diameter of the
clusters of level $i$.

\begin{lemma}
\label{lem:cluster_size}
The size (number of vertices) of a cluster of rank $j$ at any level is
at least $n^{\frac{j}{k}}$.
\end{lemma}

\begin{proof}{Lemma~\ref{lem:cluster_size}}
We prove the claim using induction on $j$. For $j = 0$, the claim
follows trivially as each cluster of rank $0$ has size at least
$1$. For the induction step, observe that a cluster of rank $j$
contains more than $n^{\frac{1}{k}}$ clusters of rank $j-1$ which all
have size at least $n^{\frac{j-1}{k}}$ by the induction hypothesis.
\end{proof}

\begin{corollary}
\label{cor:rank}
At any level of the hierarchy, The rank of a cluster can at most be
$k-1$.
\end{corollary}

\begin{proof}{Corollary~\ref{cor:rank}}
From Lemma~\ref{lem:cluster_size}, it follows that at any level of the
hierarchy there can be at most $\frac{n}{n^{\frac{k-1}{k}}} =
n^{\frac{1}{k}}$ clusters of rank $k-1$ which immediately implies that
no cluster of rank $k$ ever gets formed.
\end{proof}

\begin{lemma}
\label{lem:cluster_diameter}
Fix $\epsilon > 0$. The strong diameter of every cluster of level $i$
and rank $j$ is at most $\gamma^i ((\frac{4}{3} + \epsilon)4^j -
\frac{4}{3})$, provided $\gamma \geq \frac{1}{\epsilon}((\frac{4}{3} +
\epsilon)4^{k-1} - \frac{4}{3})$.
\end{lemma}

\begin{proof}
We prove the claim by induction on $i$ and $j$. The case for $i=-1$ is
trivially true. For the case of $i \geq 0$, assume the claim to be
true for clusters of all rank at level $i-1$. Since a cluster of rank
$0$ at level $i$ is simply one of these clusters, its diameter is
bounded by $\gamma^{i-1} ((\frac{4}{3} + \epsilon)4^{k-1} -
\frac{4}{3})$ by the induction hypothesis Corollary~\ref{cor:rank}.
This is at most $\gamma^i((\frac{4}{3} + \epsilon)4^0 - \frac{4}{3}) =
\gamma^i \epsilon$ by our assumption that $\gamma \geq
\frac{1}{\epsilon}((\frac{4}{3} + \epsilon)4^{k-1} -\frac{4}{3})$
which proves the claim for level $i$ and rank $0$.

Now assume that the claim is true for level $i$ and all rank at most
$j-1$, and consider a cluster $C$ at level $i$ and rank $j$. There are
two cases to consider depending upon whether $C$ was formed in the
first or second phase of stage $j$.

If $C$ was formed in the first phase, then there was a vertex $v$ in a
cluster of rank at most $j-1$ such that $C$ is the union of the
cluster containing $v$ and all the clusters of rank $j-1$ that the
ball $B(v,\gamma^i)$ intersected. By the induction hypothesis, the
strong diameters of all these clusters which were merged to form $C$
are bounded by $\gamma^i ((\frac{4}{3} + \epsilon)4^{j-1} -
\frac{4}{3})$. This implies that any vertex in $C$ is at most a
distance $\gamma^i ((\frac{4}{3} + \epsilon)4^{j-1} - \frac{4}{3}) +
\gamma^i$ from $v$. Thus the strong diameter of $C$ is at most
$2\gamma^i ((\frac{4}{3} + \epsilon)4^{j-1} - \frac{4}{3} + 1) \leq
\gamma^i ((\frac{4}{3} + \epsilon)4^j - \frac{4}{3})$ as $j \geq 1$.

If $C$ was formed in the second phase, it implies that there was a
cluster $C'$ of rank $j$ which was formed in the first phase of stage
$j$ and got merged with other clusters to form $C$ in the second
phase. By the argument above, the strong diameter of $C'$ was at most
$2\gamma^i ((\frac{4}{3} + \epsilon)4^{j-1} - \frac{4}{3}) +
1)$. Furthermore, we know that any vertex in $C$ either comes from
$C'$ or from some cluster of rank $j-1$ which intersects the ball
$B(v,\gamma^i)$ for a vertex $v$ contained in $C'$. From the above
facts and the induction hypothesis, we conclude that the strong
diameter of $C$ is bounded by $2\gamma^i ((\frac{4}{3} +
\epsilon)4^{j-1} - \frac{4}{3}) + 1) + 2\gamma^i +
2\gamma^i((\frac{4}{3} + \epsilon)4^{j-1} - \frac{4}{3}) =
\gamma^i((\frac{4}{3} + \epsilon)4^j - \frac{4}{3})$.
\end{proof}

\begin{LabeledProof}{Theorem~\ref{thm:hierarchical_partition}}
The bound on cluster diameter is given by
Lemma~\ref{lem:cluster_diameter} and Corollary~\ref{cor:rank}. For the
intersection bound, observe that for any level $i$ of the hierarchy
and any vertex $v$, the ball $B(v,\gamma^i)$ can intersect at most
$n^{\frac{1}{k}}$ clusters of a given rank.  This implies that
$B(v,\gamma^i)$ can intersect at most $kn^{\frac{1}{k}}$ clusters in
total from level $i$ as every cluster has rank between $0$ and $k-1$.
\end{LabeledProof}

\section{Proofs for Section~\ref{sec:cluster}}
\label{proofs:cluster}
\begin{algorithm}[ht!]
\caption{The Cluster Aggregation algorithm}
\label{alg:merger}
\begin{algorithmic}[1]
  \REQUIRE An undirected graph $G$, partition ${\cal P}$, set $S$ of
  portals.  

  \ENSURE A cluster aggregation $\cntr$

  \STATE For each set $X$ in $\partn$, let $p_X$ denote a shortest
  path from $X$ to $S$, and let $P_X$ denote the sequence of clusters
  visited in $p_X$.

  \STATE For a cluster $Y$ that appears in $P_X$, define the {\em
    position}\/ of a cluster $Y$ in $P_X$ to be $\ell$ if the number
  of distinct clusters that $P_X$ visits before first visiting $Y$ is
  $\ell-1$.

  \STATE Construct an auxiliary directed graph $D$ whose vertices are
  the clusters of $\pn$.  For vertices $X$ and $Y$, $D$ has an edge
  from $X$ to $Y$ if $P_X$ contains $Y$; furthermore, we label the
  edge $(X,Y)$ with the position of $Y$ in $P_X$.

  \STATE Set $i$ to be $0$ and $V_0$ to be the set of vertices in $D$.

  \REPEAT[Begin Phase~$i$]

  \STATE Let $D_i$ denote the subgraph of $D$ induced by $V_i$.  Let
  $E_i$ denote the set of edges in $D_i$.  Set $V_{i+1}$ to
  $\emptyset$ and $\widehat D$ to $D_i$.

  \REPEAT

  \STATE Let $v$ be an arbitrary vertex in $\widehat{D}$.  

  \IF{$i = 0$}
  
  \STATE Set $\cntr(v)$ to be the vertex in $S$ nearest to $v$;

  \ELSE
  
  \STATE Set $\cntr(v)$ to be $\cntr(x)$ where $x$ is a vertex in
  $V_{i-1} - V_i$ and the label of $(v,x)$ is the least among
  all edges from $v$ to $V_{i-1} - V_i$.  

  \ENDIF

  \STATE Let $T$ denote $\{v\} \cup \outE{\hH}{\{v\}}$.

  \REPEAT[iteration]

  \STATE
  For each $u$ in $\hH - T$, and each edge $(u,w)$ in $\hH$,
  remove $(u,w)$ from $\hH$ if there exists an edge $(u,x)$ in $\hH$
  with $x \in T$ such that the label of $(u,x)$ is smaller than the
  label of $(u,w)$.

  \STATE For each $u$ in $\inE{\hH}{T} \cup \outE{\hH}{T \cup
    \inE{\hH}{T}}$, set $\cntr(u)$ to be equal to $\cntr(v)$.  Set $T$
  equal to $T \cup \inE{\hH}{T} \cup \outE{\hH}{T \cup \inE{\hH}{T}}$.
  
  \UNTIL{$|\inE{\hH}{T}| < |T|$}.
    
  \STATE Set $V_{i+1}$ to $V_{i+1} \cup \inE{\hH}{T}$ and remove
  $T \cup \inE{\hH}{T}$ from $\hH$.

  \UNTIL{$\hH$ is empty}

  \STATE Increment $i$ \COMMENT{End Phase~$i$}

  \UNTIL{$V_{i}$ is $\emptyset$} 

\end{algorithmic}
\end{algorithm}

\begin{LabeledProof}{Lemma~\ref{lem:phase}}
We first note that $V_{i+1} \subseteq V_i$.  Furthermore, in each
iteration of the $i$th phase, when we add $\inE{\hH}{T}$ to $V_{i+1}$,
$|\inE{\hH}{T}|$ is less than $|T|$, where $T$ is a subset of $V_i -
V_{i+1}$.  Thus, $|V_i| - |V_{i+1}| \ge |V_{i+1}|$, yielding the
desired claim.
\end{LabeledProof}

For each $r_i$ in $S$, let $C(r_i)$ denote the union of the clusters
$X$ such that $\cntr(X) = r_i$.  Note that $C(r_i)$ may vary as the
algorithm progresses.

\begin{LabeledProof}{Theorem~\ref{thm:max-detour}}
Let $m$ equal $|\partn|$, the number of clusters in $\partn$.  Fix a
portal $r$ in $S$.  We will show that at the end of iteration $j$ of
phase~$i$, the following holds:
\begin{itemize}
\item
For any $Z$ in $\partn$, if $\cntr(Z)$ equals $r$, then for each
vertex $v$ in $Z$, there is a path in $\subG{G}{C(r)}$ from $v$ to
$\cntr(Z)$ of weight at most $2((i-1) \log (|\partn|) + j)
\MaxDiam(\partn)$ more than $\dist{Z}{S}$.
\end{itemize}
Before we establish the above claim, we show how the statement of the
theorem follows.  By Lemma~\ref{lem:phase}, the number of phases is at
most $\log m$.  Furthermore, the number of iterations of the inner
repeat loop in each phase is at most $\log m$ since the size of $T$ at
least doubles in each iteration.  Therefore, at termination, the
detour for each cluster in $\partn$ is at most $2(\log^2 m)
\MaxDiam(\partn)$, yielding the desired claim.  \junk{ subgraph
  induced by $C(r)$ for any center $r$ of $S$ has diameter at most
  $\gamma^i + 2\alpha \log^2 n \gamma^{i-1} \le \alpha \gamma^i$, if
  $\alpha$ and $\gamma$ satisfy $\gamma + 2\alpha \log^2 n \le \alpha
  \gamma$; setting $\gamma \ge 4 \log^2 n$ and $\alpha \ge 2$
  suffices.  }

Consider an iteration $j$ of phase $i$.  In the following, $T$ and
$\hH$ refer to the variables in the above algorithm at the start of
the iteration.  The set of clusters for which we set the $\cntr$
values in the iteration is given by $\inE{\hH}{T} \cup \outE{\hH}{T
  \cup \inE{\hH}{T}}$, where $T$ corresponds to the value of the
variable at the start of the iteration.  Every cluster in $T$ shares
the same $\cntr$ value, say $x$.  By the induction hypothesis, at the
start of iteration $j$ of phase $i$, each cluster $Y$ in the set of
clusters with $\cntr$ equal to $x$ has a path $q_Y$ in
$\subG{G}{C(x)}$ from $Y$ to $x$ of weight at most $2((i-1)\log m +
(j-1))\MaxDiam(\partn)$ more than $\dist{Y}{S}$.

Consider a vertex $Z$ in $\inE{\hH}{T}$.  Since $Z$ is in
$\inE{\hH}{T}$, its path $p_Z$ contains a cluster $Y$ in $T$.  Let
$p'$ denote the prefix of the path $p_Z$ that connects $Z$ to the
first occurrence of $Y$ in $p_Z$; and let $p''$ denote the remainder
of the path $p_Z$.  We note that every cluster that appears in $p'$ is
in $\outE{\hH}{\{z\}}$, and is, hence, also in $\outE{\hH}{T \cup
  \inE{\hH}{T}}$.  Thus, at the end of iteration $j$, $p'$ is fully
contained in $\subG{G}{C(x)}$ the subgraph of $G$ induced by the set
of vertices with $\cntr$ equal to $x$.  The weight of $p_Z$ equals the
sum of the weights of $p'$ and $p''$.  The weight of $p_Y$ is at most
the weight of $p''$.  Thus, the path from $Z$ to $x$ consisting of
$p'$, followed by a shortest path to $p_Y$ in $Y$, and followed by the
path $q_Y$ is entirely contained in $\subG{G}{C(x)}$ and has weight at
most $2((i-1)\log m + j)\MaxDiam(\partn)$ more than the length of
$p_Z$.  (This is because the weight of any shortest path in $Y$ is at
most $\MaxDiam(\partn)$.)  This completes the induction step of the
proof.
\end{LabeledProof}

\section{Proofs for Section~\ref{section:minor-free}}
\label{proofs:minor-free}
\begin{algorithm}[ht!]
\caption{Clustering of minor-free graph connected component}
\label{alg:minor-separator}
\begin{algorithmic}[1]
  \REQUIRE Connected component $\Phi$ of minor-free graph $G$,
  strong $(\alpha, \beta, \gamma^{i-1})$-partition $\partn_{i-1}$ of $G$,
  set $\cN$ with coarsen clusters of $\partn_{i-1}$.

  \ENSURE Coarsening the $\partn_{i-1}$ clusters in $\Phi$;
          the resulting clusters are inserted into $\cN$.

  \STATE Let $S = S_{1} \cup S_{2} \cup \cdots \cup S_{l}$
         be a $k$-path separator of $\Phi$.

  \FOR{$\chi$ from $1$ to $l$}

  \FOR{each path $p \in S_{\chi}$}

  \STATE  Let $\Psi$ be the connected component 
          of $\Phi \setminus \bigcup_{1 \leq j < \chi} S_j$ 
          in which $p$ resides.
  
  \STATE Invoke Algorithm~\ref{alg:path-clustering} on connected component $\Psi$ and path $p$,
              and parameters $\partn_{i-1}$ and $\cN$.
  
  \STATE Update $\cN$ to be the result of Algorithm~\ref{alg:path-clustering} . 

  \ENDFOR

  \ENDFOR
  
  \FOR{each connected component $\Upsilon \in \Phi \setminus S$}
  
  \STATE Invoke (recursively) Algorithm \ref{alg:minor-separator} 
         with parameters  $\Upsilon$, $\partn_{i-1}$, and $\cN$.
         
  \STATE Update $\cN$ to be the result of the recursive invocation. 
  
  \ENDFOR
  
  \STATE Return $\cN$.
\end{algorithmic}
\end{algorithm}

\begin{algorithm}[ht!]
\caption{Path clustering in connected component}
\label{alg:path-clustering}
\begin{algorithmic}[1]
  \REQUIRE Connected component $\Psi$ of minor-free graph $G$,
  path $p$ in $\Psi$,
  strong $(\alpha, \beta, \gamma^{i-1})$-partition $\partn_{i-1}$ of $G$,
  set $\cN$ with coarsen clusters of $\partn_{i-1}$.

  \ENSURE Coarsening the $\partn_{i-1}$ clusters in $\Psi$  which are at distance at most $2^i$ from  $p$;
          the resulting clusters are inserted in $\cN$.

  \STATE Let $\partn^\Psi_{i-1} = \{X \in \partn_{i-1}: X \subseteq V(\Psi)\}$
         be the {\em integral} clusters of $\partn_{i-1}$
         which are completely contained within $\Psi$; 

  \STATE Let $\cA = \{X \in \partn^\Psi_{i-1}: 
         (\mydist_{\Psi}(X, p) \leq 2 \gamma^i) 
         \wedge (X \cap V(\cN) = \emptyset)\}$
         be the all integral clusters of $\Psi$ which have not yet 
         been coarsen (do not belong in  $\cN$) and are
         within distance $2 \gamma^i$ from $p$ in $\Psi$. 

  \STATE Let $\cB = \{X \in \cA:
		 \exists (u,v) \in E(\Psi), 
		 u \in X \wedge v \in V(\cN) 
		 \cup (V(\Psi)\setminus V(\partn^\Psi_{i-1}) \}$,
		 contains all the clusters of $\cA$
		 which are adjacent to clusters in $\cN$
		 or adjacent to non-integral clusters in $\Psi$. 

  \STATE Let $\Psi' = \Psi \cap V(\cA)$ be the sub-graph of $\Psi$ induced 
         by $V(\cA)$ (note that $\Psi'$ may not be connected).      
         
  \STATE Let $L$ be the {\em leaders} of path $p$, 
		 which is a maximal set of nodes in $p \cap \Psi'$,
		 such for any pair $u,v \in L$,
		 $\mydist_p(u,v) \geq \gamma^i$,
		 and $u$ and $v$ cannot belong to the same cluster of $\cA$.

  \STATE Let $U$ be the set that consists of one arbitrary node from each cluster in  $\cB$
               (for each cluster in $\cB$ that does not contain a leader in $L$).
 
  \STATE Combine the clusters in $\cA$ 
         by invoking Algorithm \ref{alg:merger}
         to each connected component of $\Psi'$ for the induced clusters 
         from $\cA$ and the induced portal nodes in $L \cup U$.
         
  \STATE Let $\cR$ be the union of the resulting set of clusters 
         from Algorithm \ref{alg:merger}.
         
  \STATE Write $\cR = \cI_p \cup \cK_p$ where
         $\cI_p$ consists of clusters that contain a node of $L$,
         and $\cK_p$ consists of clusters that contain a node of $U$.
         
  \FOR{each cluster $X \in \cK_p$ 
            such that $X$ is adjacent to a cluster $Y \in \cN$
           such that there is an edge $(u,v) \in E(\Psi)$,
           $u \in X$, $v \in Y$, and $v \notin \Psi'$}

  \STATE $X$ merges with $Y$ and the new cluster is inserted back in $\cN$.
   
  \STATE Remove $X$ from $\cK_p$. 
     
  \ENDFOR  
  
  \STATE Update $\cN = \cN \cup \cI_p \cup \cK'_p$, where $\cK'_p$ are the remaining clusters of $\cK_p$.
  
  \STATE Return $\cN$.
\end{algorithmic}
\end{algorithm}

Consider a node $t \in T$.
Each path $p \in S(t)$ has a respective processing order
in $S(t)$, denoted $order(p)$, which is a unique
integer between 1 and $k$.
The set of {\em previous} paths of $p$, denoted $Q(p)$,
is defined to include those
paths in $S(t)$ which have smaller order,
or the paths in the ancestors of $t$:
$$Q(p) = \{q \in S(t): order(q) < order(p) \} 
\cup \{q \in S(w): \mbox{$w$ is ancestor of $t$}\}.$$
Let $\cI_p$ denote the clusters that
belong to $p$ immediately after $p$ is processed.  Let $\wI_p$ denote
the final clusters of $p$ in $\partn_{i}$.  In the analysis below, we
assume that $\gamma \ge \alpha$, which we can satisfy in our
construction.

\begin{lemma}
\label{thm:minor-distance}
In $\Psi'$ every cluster of $\cA$
is within distance at most $3\gamma^i$ to a node in $L \cup U$. 
\end{lemma}

\begin{proof}
Consider a cluster $X \in \cA$.
Let $u \in X$ be the closest to a node $v \in p$ in graph $\Psi$.
From the definition of $\cA$,
$\mydist_{\Psi}(u, v) \leq 2 \gamma^i$.
Let $q$ be a shortest path in $\Psi$ connecting $u$ to $v$.

If $q$ uses a cluster outside $A$,
then that cluster must be either a cluster in $\cN$
or a non-integral cluster of $\partn_{i-1}$.
Therefore, $q$ has to cross a cluster in $\cB$.
Let $\ell \in V(\cB) \cap U$.
Since $\alpha \gamma^{i-1} \leq \gamma^i$, 
$\mydist_{\Psi'}(u, \ell) \leq
2 \gamma^i + \alpha \gamma^{i-1} \leq 3 \gamma^i$.

Consider now the case where $q$ uses only clusters in $A$.
Let $p'$ be the subpath of $p$ which consists
of the nodes within distance $\gamma^{i}$ from $u$,
with respect to $\Psi$.

Suppose that $p'$ uses only clusters in $\cA$.
For the sake of contradiction,
assume that none of the nodes in $p'$ is a leader in $L$.
Let $Y \in \cA$ be the cluster that contains $v$.
We have that the closest leader to $u$ (if it exists),
must be at distance greater than $\gamma^i$
from $v$.
Since the diameter of $Y$ is at most 
$\alpha \gamma^{i-1} \leq \gamma^{i}$,
then $L$ is not maximal because 
$v$ is a valid possible leader.
Therefore, $p'$ must contain a leader $\ell \in L$.
Thus, 
$\mydist_{\Psi'}(u, \ell) \leq
2 \gamma^i + \gamma^{i} \leq 3 \gamma^i$.

If $p'$ doesn't use a cluster in $\cA$,
then it has to use a cluster in $\cB$.
By selecting a node $\ell \in V(\cB) \cap U$,
we get 
$\mydist_{\Psi'}(u, \ell) \leq
2 \gamma^i + \alpha \gamma^{i-1} \leq 3 \gamma^i$.
\end{proof}

\begin{lemma}
\label{thm:minor-diameter}
Every cluster of $\wI_p$ has diameter at most 
$\alpha' \gamma^i$, where $\alpha' =  c_1 k \log^3 n$ for some positive constant $c_1$.
\end{lemma}

\begin{proof}
From Lemma \ref{thm:minor-distance},
each cluster in $\cA$ is within distance $3 \gamma^i$
from a portal node in $L \cup U$.
Since each cluster in $\cA$ has diameter at most $\alpha \gamma^{i-1} \leq \gamma^i$,
each node in $V(\cA)$ 
is within distance at most $4 \gamma^i$
from a portal node.
Algorithm \ref{alg:merger} merges the clusters in $\cA$ to produces new clusters in $\cI_p$ and $\cK_p$ 
where each cluster is ``centered'' at a node in $L \cup U$.
From Theorem \ref{thm:max-detour}, the detour of each node in its new cluster  
is at most $O((\alpha \gamma^{i-1}) \log^2 |\cA|) = O(\gamma^{i} \log^2 n)$.
Therefore, the distance of each node to the center node is at most $4 \gamma^i + O(\gamma^{i} \log^2 n)$.
Thus, the diameter of the new clusters is at most twist the distance of its nodes to the center,
namely, at most $\zeta = c \gamma^i \log^2 n$, for some appropriately chosen constant $c$.

The clusters in $\cI(p)$ may increase in diameter,
when they merge with $\cK_q$ clusters from any path $q$
processed after $p$.
Path  $q$ may belong to $S(t)$ or to $S(w)$, 
where $w$ is a descendant in the sub-tree $T' \subseteq T$ rooted at $t$.
Each path $q \in S(t)$ with order after $p$, 
increases the diameter of $\cI_p$ by at most $2 \zeta$,
since newly merged clusters from $\cK_q$
add at most one layer of clusters into $\cI_p$,
and any two clusters in the layer can reach each other
through the previous instance of $\cI_p$.
Thus, when we process the last path in $S(t)$,
we have added at most $k$ layers,
and the increase in the diameter of the new $\cI_p$ 
will be at most $2 \zeta k$.

Similarly, any node in the sub-tree $T'$,
contributes at most $k$ new layers to $\cI_p$.
However, all the nodes of $T'$ in the same level
contribute in total $k$ layers,
since clusters in them
are formed independent of each other.
Since the sub-tree $T$ has at most $1 + \log n$ levels (including $t$),
we have in total $k (1 + \log n)$ additional layers in $\cI_p$,
contributing increase at most $2 \zeta k (1 + \log n)$ to the diameter
of $\cI(p)$.
Therefore,
the diameter of $\wI(p)$ is at most 
$2 \zeta k (1 + \log n) + \zeta \leq c_1 k \gamma^i \log^3 n$,
for some constant $c_1$.     
\end{proof}

The coarsen clusters $\cI_q$ of a path $q \in Q(p)$
may change after processing $p$,
since some clusters in the connected component $\Psi$ of $p$
may merge with the existing clusters of $q$. 
Let $\cI'_q$ and  $\cI''_q$ be the respective clusters of $q$ 
just before and after processing path $p$.
Let $\cZ'(p) = \bigcup_{q \in Q(p)} \cI'_q$ and $\cZ''(p) = \bigcup_{q \in Q(p)} \cI''_q$
be the sets which consist of all the new coarsen clusters of the paths in $Q(p)$
before and after we process $p$, respectively.
Let $\cZ(p) = \cI_p \cup \cZ''(p)$ be the set of coarsen clusters that have been
formed so far by the paths in $p \cup Q(p)$.
We observe that
for any path $q \in Q(p)$ it holds that $V(Z(q)) \subseteq V(Z'(p)) \subseteq V(Z(p))$,
since the previous clusters of $p$ may only grow before $p$ is processed. 

For any set of nodes $Y$ denote with
$\Gamma(Y)$ the set of clusters in $\partn_{i-1}$
which are within distance $2 \gamma^i$ from $p$,
namely,
$\Gamma(Y) = \{X \in \partn_{i-1} : \mydist_G(X, Y) \leq 2 \gamma^i\}$.
Note that $\Gamma(p)$ contains both the integral and non-integral clusters of $\Psi$ at distance 
upto $2 \gamma^i$ from $p$.
In the next result we show that each cluster in $\Gamma(p)$ must be included 
in some coarsen cluster of $p \cup Q(p)$.
This result also implies that each node in path $p$ will be a member of some cluster 
which either belongs to $\wI_p$ or to $\wI_q$ of a path $q \in Q(p)$.

\begin{lemma}
\label{thm:minor-biggamma}
$\Gamma(p) \subseteq \cZ(p)$.
\end{lemma}

\begin{proof}
We prove the claim by induction on $|Q(p)|$.
For the basis case, $|Q(p)| = 0$,
path $p$ is the first to be processed by the algorithm
with $Q(p) = \emptyset$.
Therefore, $\Gamma(p) = \cI_p = \cZ(p)$.
Assume now that the claim holds for $|Q(p)| \leq \sigma$,
and consider the case $|Q(p)| = \sigma + 1$.
From induction hypothesis,
for each path $q \in Q(p)$,
$\Gamma(q) \subseteq \cZ(q)$.

Let $\cN$ be the set of newly formed coarsen clusters of the algorithm 
just before we process $p$.
First, we show
that just before we process path $p$
the clusters of $\cN$ that intersect $\Psi$
can only be those in $\cZ'(p) \subseteq \cN$. 
Suppose, for the sake of contradiction,
that there is a cluster $X \in \cN \setminus \cZ'(p)$
which intersects $\Psi$.
Cluster $X$ must be a non-integral cluster in $\Psi$, 
namely, contains a node $y \notin V(\Psi)$,
since any integral cluster in $\Psi$ can only have been 
built by a path in $Q(p) \cap S_\chi$,
where $p \in S_\chi$.  
Take a node $u \in X  \cap V(\Psi)$.
Any path in $X$ from $u$ to $y$ must cross some path $q \in Q(p)$
whose removal from $G$ contributed to the formation of $\Psi$. 
Since $q \in V(\Gamma(q))$,
and from induction hypothesis $\Gamma(q) \subseteq \cZ(q)$,
we have that $q \in V(\cZ(q)) \subseteq V(\cZ'(p))$.
Thus, $X$ has to be a cluster of $\cZ'(p)$, a contradiction.

Next, we show that any non-integral 
cluster $Y \in \partn_{i-1}$, $Y \notin \partn^\Psi_{i-1}$,
which intersects $\Psi$ is used in a cluster of $\cZ'(p)$.
Note that $Y$ must have been crossed
by at least a path $q \in Q(p)$ whose removal 
from $G$ contributed to the creation of $\Psi$.
Since the diameter of $Y$ is bounded by 
$\alpha \gamma^{i-1} \leq \gamma^{i}$,
we have that $Y \in \Gamma(q) \subseteq \cZ(q)$.
Therefore, $Y \in \cZ'(p)$, since $V(\cZ(q)) \subseteq V(\cZ'(p))$.

We continue now with the main claim.
Consider a cluster $X \in \Gamma(p)$.
There are the following possibilities:
\begin{itemize}
\item
$X \in \partn^{\Psi}_{i-1}$:
$X$ is integral in $\Psi$ and we examine the following sub-cases.

\begin{itemize}
\item
$X \in \cN$:
since before processing $p$ 
only clusters in $\cZ'(p)$ intersect $\Psi$,
we get $X \in \cZ'(p)$.
Therefore, according to the algorithm,
after processing $p$ cluster $X$ will remain in the same cluster as in $\cZ'(p)$.
Thus, $X \in \cZ(p)$.

\item
$X \in \cA$:
from the algorithm, after processing $p$ there are two possibilities.
First possibility is $X \in \cI_p$, and hence, $X \in \cZ(p)$. 
Second possibility is $X \in \cK_p \setminus \cI_p$
and $X$ is either (i) adjacent to some node in $\cN$,
or (ii) adjacent to some non-integral cluster in $\Psi$.
In case (i) $X$ merges with a cluster in $\cN$,
and since only clusters of $\cZ'(p) \subseteq \cN$ can be in $\Psi$,
we immediately have $X \in \cZ(p)$.
In case (ii),
as we have shown above any non-integral cluster
of $\partn_{i-1}$ in $\Psi$
is a member of $\cZ'(p) \subseteq \cN$,
and thus $X$ merges with a cluster of $\cZ'(p)$,
which implies that $X \in \cZ(p)$.
\end{itemize}

\item
$X \notin \partn^{\Psi}_{i-1}$:
$X$ is either non-integral in $\Psi$ or does not intersect $\Psi$ at all.
Then, $X$ must contain a node $u \notin \Psi$.
If $X$ intersects $\Psi$,
then we have shown above that $X \in \cZ'(p)$,
and thus, $X \in \cZ(p)$.  
If $X$ does not intersect $\Psi$,
any path from $p$ to $X$ must intersect a path $q \in Q(p)$,
since otherwise $X$ wouldn't reside in a different component than $\Psi$.
Since $\mydist_G(p, X) \leq 2 \gamma^i$,
we have that $\mydist_G(q, X) \leq 2 \gamma^i$.
Therefore, $X \in \Gamma(q) \subseteq \cZ(q)$.
Consequently, $X \in \cZ(p)$.
\end{itemize} 
\end{proof} 

\begin{lemma}
\label{thm:minor-Iintersection}
Any ball of radius $\gamma^i$ in $G$
intersects with at most $2\alpha'+3$ 
clusters of $\wI_p$.
\end{lemma}

\begin{proof}
We start by showing that we only need to consider 
balls of radius $\gamma^i$ in $\Psi$.
Let $G' = G \setminus \Psi$.
Let $Y$ denote the set of nodes in $G'$ such that
each $x \in Y$ is adjacent to a node in $\Psi$.
It must be that each $x$ is a member of a path in $Q(p)$,
since $x$ is removed from the network before path $p$ is processed.
In other words, $Y \subseteq V(Q(p))$.

Let $\cF$ be all the (integral) clusters in $\partn_{i-1}^\Psi$
which are at distance at most $2 \gamma^i$ from $Y$,
namely, 
$\cF = \{X \in \partn_{i-1}^\Psi : \mydist_\Psi(X, Y) \leq 2 \gamma^i\}$.
Clearly, 
$\cF = \partn_{i-1}^\Psi \cap \Gamma(Y)$.

If we apply Lemma \ref{thm:minor-biggamma} to any path $q \in Q(p)$,
we obtain that $\Gamma(q) \subseteq \cZ(q)$.
Since $\Gamma(Y) \subseteq \bigcup_{q \in Q(p)} \Gamma(q)$,
we obtain  $\Gamma(Y) \subseteq \bigcup_{q \in Q(p)} \cZ(q)$.
Since $ V(\cZ(q)) \subseteq V(\cZ'(p))$, we obtain $\Gamma(Y) \subseteq \cZ'(p)$.
Therefore, the clusters in $\cF$ 
are all used in coarsen clusters of paths in $Q(p)$ just before processing $p$,
that is, $\cF \subseteq \cZ'(p)$.
According to the algorithm, the coarsen clusters of $p$, $\cI_p$, 
cannot possibly contain any of the coarsen clusters in $\cZ'(p)$,
namely, $\cI_p \cap \cZ'(p) = \emptyset$.
Consequently, $\cI_p \cap \cF = \emptyset$.
When the algorithm further processes the remaining paths in $\Psi$ of path set $S(t)$ 
(paths ordered after $p$ in $S(t)$),
and then the descendants of $t$ in $T$,
we have that each of the coarsen clusters in 
$I_p$ may grow (including new clusters of $\partn_{i-1}$), 
however, they will never intersect $\cF$.
Thus, $\wI_p \cap \cF = \emptyset$.

Consequently,
any cluster of $\wI_p$ is at distance at least $2 \gamma^i$
from $G'$.
Therefore, any ball of radius $\gamma^i$
that intersects clusters of $\wI_p$ has to be
a sub-graph of $\Psi$.
Thus, in order to prove the main claim,
we only need to focus on graph $\Psi$.

Consider now a ball $B = B(u,\gamma^i)$ within $\Psi$.
Suppose that $\xi \geq 2$ clusters of $\wI_p$
intersect $p$. 
Path $p$ is a shortest path in $\Psi$.
Each cluster in $\wI_p$ 
has a distinct leader in $p$.
The leaders are at distance at least $\gamma^i$ apart in $p$.
Therefore, there are two clusters intersecting $B$,
whose respective leaders, $\ell_1$ and $\ell_2$, 
are at distance at least $\mydist_\Psi(\ell_1, \ell_2) \geq (\xi - 1) \gamma^i$.
Ball $B$ provides an alternative
path between $\ell_1$ and $\ell_2$ through $u$,
with total length is bounded by 
$\mydist_\Psi(\ell_1, \ell_2) \leq \mydist_\Psi(\ell_1, u) + \mydist_\Psi(u, \ell_2)$.
Since the cluster of $\ell_1$ intersects $B$,
we obtain from Lemma \ref{thm:minor-diameter} that
$\mydist_\Psi(\ell_1, u) \leq \alpha' \gamma^i + \gamma^i = (\alpha' + 1) \gamma^i$.
Similarly, 
$\mydist_\Psi(u, \ell_2) \leq (\alpha' + 1) \gamma^i$.
Therefore,
$\mydist_\Psi(\ell_1, \ell_2) \leq 2(\alpha' + 1) \gamma^i$.
Therefore, it has to be $\xi - 1 \leq 2 (\alpha' + 1)$,
or equivalently,
$\xi \leq 2 \alpha' + 3$. 
\end{proof}

\begin{lemma}
\label{thm:minor-intersection}
Any ball of radius $\gamma^i$ in $G$
intersects with at most $c_2 \alpha' k \log n$
clusters of $\partn_{i}$,
for a constant $c_2$.
\end{lemma}

\begin{proof}
Consider a node $v \in G$ and the ball $B = B(v, \gamma^i)$.
Each node $v \in G$ belongs to at least one path in a path separator in $T$. 
Let $p \in S(w)$ be 
the first path to be processed by the algorithm with $v \in p$.
Clearly, $B(v,\gamma^i) \subseteq \Gamma(p)$.
Therefore, from Lemma \ref{thm:minor-biggamma},
we have that $B \subseteq \cZ(p)$.
Since $\cZ(p)$ consists only of clusters that belong
to $Q' = p \cup Q(p)$,
all the paths in $Q'$ appear in path separators of $T$ 
between the root and $w$. 
Since the depth of $T$ is at most $1 + \log n$,
the total number of path separators
involved in $Q'$ is at most $1 + \log n$,
each contributing $k$ paths.
Therefore,
$|Q'| \leq k (1 + \log n)$.
   
From Lemma \ref{thm:minor-Iintersection},
$B$ intersects with at most $(2 \alpha' + 3)$
clusters of each path $q \in Q'$.
Thus,
the total number of clusters of $\partn_{i}$ intersecting $B$
is at most $(2 \alpha' + 3) k (1 + \log n) \leq c_2 \alpha' k \log n$,
for a constant $c_2$, as needed.  
\end{proof}   


\begin{LabeledProof}{Lemma~\ref{thm:minor-partn_i}}
Every node in $G$ belongs to a path in some path separator
used by the algorithm.
From Lemma \ref{thm:minor-biggamma},
each node in a path $p$ must be a member of some coarsen  cluster 
which either belongs to $p$ or to a path $q \in Q(p)$.
Thus, for each path $p$, $\cK'_p = \emptyset$.
Consequently,
each node $v \in G$ will eventually appear in some cluster $\wI_q$
of some path $q \in S$.
Therefore, $\partn_{i}$ is a partition of $G$.

From Lemma \ref{thm:minor-diameter},
the diameter of any $\wI_q$ is bounded by $\alpha' \gamma^i$.
Therefore, 
the diameter of each cluster 
in $\partn_{i}$ is at most $\alpha' \gamma^i$.
In addition, from Lemma \ref{thm:minor-intersection},
each ball of radius $\gamma^i$ 
intersects at most $c_2 \alpha' k \log n$
clusters of $\partn_{i}$.
Consequently, $\partn_{i}$ is a $(\alpha', c_2 \alpha' k \log n , \gamma^i)$-partition of $G$.
\end{LabeledProof}


\begin{LabeledProof}{Theorem~\ref{thm:minor-universal}}
From Lemma \ref{thm:minor-partn_i},
since in fixed minor-free graphs $k = O(1)$,
we can build a hierarchy of clusters 
by choosing $\alpha = \alpha' = O(\log^3 n)$.
Further, for each level $i$,
we can create the necessary padding around 
a root node $r \in G$ of radius $\gamma^i$,
by creating a cluster that contains the ball $B(r, \gamma^i)$.
We can do this by using either of two methods.
In the first method, we can explicitly add $r$ to the first separator in $G$
as an artificial path (with one node) that needs to be processed first.
This causes the size of the first separator to be of size $k + 1$,
and in the analysis we replace $k$ with $k+1$.
In the second method, we can merge all the clusters in $B(r, \gamma^i)$
created by the algorithm, giving a new cluster whose diameter 
is no more than three times the diameter of the old cluster.
Either way, the impact to the parameters of the 
clustering is a constant factor,
giving the desired hierarchical partition.
It is easy to verify that all the steps of the algorithm 
can be performed in polynomial time with respect to the size of $G$ 
and the parameters of the problem.
\end{LabeledProof}


\section{Conclusions}
\label{sec:concs}
In this paper, we have presented a polynomial-time $2^{O(\sqrt{\log
    n})}$-stretch \ust\ construction for general graphs, which is the
first known subpolynomial-stretch ($o(n^\epsilon)$ for any $\epsilon >
0$) solution for general graphs.  We have also presented a
$\polylog(n)$-stretch \ust\ algorithm for minor-free graphs, for which
$\Omega(\log n)$ is a known lower bound.` Both \ust\ algorithms are
based on a framework that draws close connections between a certain
class of strong graph partitions and low-stretch \ust s.  Our modular
framework leads us to designing new strong-diameter partitions for
both general and minor-free graphs, and solving a new cluster
aggregation problem, all of which are of independent interest.

Our work leaves several important open problems.  The most compelling
one is that of deriving tight bounds on the best stretch achievable
for general graphs (specifically, is $\polylog(n)$-stretch
achievable?).  For minor-free graphs, the exponent in the
$\polylog(n)$ factor we achieve for stretch is high.  Our current
analysis follows the modular algorithmic framework; we believe that an
improved bound can be achieved by a more careful ``flatter'' analysis.
Furthermore, any improved approximation for the cluster aggregation
problem will yield significant improvements in the \ust\ stretch
factors.



\bibliographystyle{plain}
\bibliography{GraphUST}

\end{document}